\documentclass[11pt]{article}
\usepackage[utf8]{inputenc}

\usepackage{amsmath, amssymb, amsthm, amsfonts}
\usepackage{mathtools}
\usepackage{thm-restate}
\usepackage{float}
  
\usepackage[dvipsnames,svgnames,x11names]{xcolor}
\definecolor{ForestGreen}{rgb}{0.1333,0.5451,0.1333}
\usepackage[colorlinks=true,linkcolor=ForestGreen,citecolor=blue]{hyperref}

\usepackage[margin=1in]{geometry}
\usepackage{graphics}
\usepackage{pifont}
\usepackage{tikz}
\usepackage{bbm}
\usepackage[T1]{fontenc}

\DeclareMathOperator*{\argmin}{argmin}
\usetikzlibrary{arrows.meta}
\usepackage{environ}
\usepackage{framed}
\usepackage{url}
\usepackage[linesnumbered,ruled,vlined]{algorithm2e}
\usepackage[noend]{algpseudocode}
\usepackage[noabbrev,capitalize,nameinlink]{cleveref}
\crefname{equation}{}{}
\usepackage[labelfont=bf]{caption}
\usepackage{cite}
\usepackage{framed}
\usepackage[framemethod=tikz]{mdframed}
\usepackage{appendix}
\usepackage{graphicx}
\usepackage[textsize=tiny]{todonotes}
\usepackage{tcolorbox}
\allowdisplaybreaks[1]

\newcommand\remove[1]{}

\newtheorem{lemma}{Lemma}[section]
\newtheorem{theorem}{Theorem}
\newtheorem*{lemma*}{Lemma}

\newtheorem*{corollary*}{Corollary}

\theoremstyle{definition}

\newtheorem*{theorem*}{Theorem}
\newtheorem{definition}[lemma]{Definition}

\newtheorem*{rem*}{Remark}

\newcommand{\eps}{\varepsilon}

\newcommand{\pe}{\preceq}
\newcommand{\one}{\mathbbm{1}}
\newcommand{\abs}[1]{\left|#1\right|}
\newcommand{\norm}[1]{\left\lVert#1\right\rVert}
\crefname{algocf}{Algorithm}{Algorithms}
\crefname{claim}{Claim}{Claims}

\renewcommand{\l}{\langle}
\renewcommand{\r}{\rangle}

\newcommand{\poly}{\mathsf{poly}}

\renewcommand{\bar}{\overline}
\renewcommand{\hat}{\widehat}

\renewcommand{\bar}{\overline}

\newcommand{\g}{\nabla}

\newcommand{\beps}{\bar{\eps}}
\newcommand{\assign}{\leftarrow}
\newcommand{\new}{\mathsf{new}}
\newcommand{\E}{\mathbb{E}}
\newcommand{\Var}{\mathsf{Var}}
\newcommand{\stab}{\mathsf{stab}}

\newcommand{\decr}{\mathsf{decr}}
\newcommand{\inte}{\mathsf{int}}
\newcommand{\final}{\mathsf{final}}

\newcommand{\tp}{\top}
\newcommand{\Otil}{\widetilde{O}}

\newcommand{\Atil}{\widetilde{A}}
\newcommand{\tautil}{\widetilde{\tau}}
\newcommand{\qhat}{\widehat{q}}

\newcommand{\albert}[1]{\textbf{\color{green}[Albert: #1]}}

\newcommand{\R}{\mathbb{R}}
\renewcommand{\O}{\widetilde{O}}
\newcommand{\that}{\widehat{t}}
\newcommand{\lev}{\mathsf{lev}}

\newcommand{\xhat}{\widehat{x}}
\newcommand{\shat}{\widehat{s}}
\newcommand{\Deltahat}{\widehat{\Delta}}
\newcommand{\xbar}{\overline{x}}
\newcommand{\sbar}{\overline{s}}

\newcommand{\epscheck}{\epsilon_{\mathrm{checker}}}
\newcommand{\epsjl}{\epsilon_{\mathrm{JL}}}
\newcommand{\epshh}{\epsilon_{\mathrm{hh}}}

\begin{document}

\title{Adaptive Matrix Sparsification and Applications to \\ Empirical Risk Minimization}

\author{
Yang P. Liu \\ Carnegie Mellon University \\ yangl7@andrew.cmu.edu
\and
Richard Peng \\Carnegie Mellon University\\ yangp@cs.cmu.edu
\and
Colin Tang \\Carnegie Mellon University\\ cstang@andrew.cmu.edu
\and
Albert Weng \\ Georgia Institute of Technology \\ albweng@gatech.edu
\and
Junzhao Yang \\ Carnegie Mellon University \\ junzhaoy@andrew.cmu.edu}

\date{}

\clearpage\maketitle

\begin{abstract}
Consider the \emph{empirical risk minimization} (ERM) problem, which is stated as follows. Let $K_1, \dots, K_m$ be compact convex sets with $K_i \subseteq \R^{n_i}$ for $i \in [m]$, $n = \sum_{i=1}^m n_i$, and $n_i\le C_k$ for some absolute constant $C_k$.
Also, consider a matrix $A \in \R^{n \times d}$ and vectors $b \in \R^d$ and $c \in \R^n$. Then the ERM problem asks to find
\[
\min_{\substack{x \in K_1 \times \dots \times K_m\\ A^\top x = b}}
    c^\top x.
\]
We give an algorithm to solve this to high accuracy in time $\widetilde{O}(nd + d^6\sqrt{n}) \le \widetilde{O}
(nd + d^{11})$~\footnote{Throughout, we use $\widetilde{O}$ to hide constants in $C_K$ as well as logarithmic dependencies in $n, d$ and the accuracy $\eps$.}, which 
is nearly-linear time in the input size when $A$ is dense and $n \ge d^{10}$.

Our result is achieved by implementing an $\widetilde{O}(\sqrt{n})$-iteration interior point method (IPM) efficiently using dynamic data structures. In this direction, our key technical advance is a new algorithm for maintaining leverage score overestimates of matrices undergoing row updates. Formally, given a matrix $A \in \R^{n \times d}$ undergoing $T$ batches of row updates of total size $n$ we give an algorithm which can maintain leverage score overestimates of the rows of $A$ summing to $\widetilde{O}(d)$ in total time $\widetilde{O}(nd + Td^6)$. This data structure is used to sample a spectral sparsifier within a robust IPM framework to establish the main result.
\end{abstract}

\newpage

\tableofcontents

\renewcommand{\Atil}{\widetilde{\mathit{A}}}
\newcommand{\Gtil}{\widetilde{\mathit{G}}}
\newcommand{\tauhat}{\widehat{\mathit{\tau}}}

\newpage

\section{Introduction}
\label{sec:intro}

Empirical risk minimization (ERM) is a general convex optimization problem which captures several fundamental tasks such as linear regression, $\ell_p$ regression \cite{C05,DDHKM08,BCLL18,AKPS19}, LASSO \cite{Tib96}, logistic regression \cite{Cox58,HLS13}, support vector machines (SVM) \cite{CV95}, quantile regression \cite{K00,KH01}, and AdaBoost \cite{FS97}.
A more comprehensive discussion of ERM is in~\cite{LSZ19}.
There, the problem is formally defined as:
\begin{equation} \min_{y \in \R^d}\sum_{i=1}^m
f_i\left(A_i y - c_i\right) \label{eq:erm1} \end{equation}
where $A_i \in \R^{n_i \times d}$, $c_i \in \R^{n_i}$ for integer dimensions $n_1, \dots, n_m$, and $f_i: \R^{n_i} \to \R$ are convex functions.
In this paper, one should think of $n_i$ being small constants and $m$ as being much larger than $d$. When all $n_i = 1$ this is known as a generalized linear model (GLM).

The ERM as stated in \eqref{eq:erm1} translates to the more
convenient form as stated in the abstract via an application of duality for convex programs.
Note each $f_i$ is convex, its convex conjugate $f_i^{*}$ is convex, 
and standard Sion's min-max duality manipulations (which we defer to \cref{sec:ERMDuality}) give
\[
\min_{y \in \R^d}\sum_{i=1}^m
f_i\left(A_i y - c_i\right)
=
\max_{x \in \R^{\sum n_i}, A^\top x = 0}
\sum_{i=1}^m -c_i^\top x_i - f_i^*\left(x_i\right)
=
- \min_{x \in \R^{\sum n_i}, A^\top x = 0}
c^\tp x + \sum_{i=1}^m f_i^*\left(x_i\right).
\]
Now introducing for each $i$ a new scalar $x^{obj}_i \in \R$
and defining the convex sets $K_i$ on $(x_i, x^{obj}_i)$ as
\[
K_i
\coloneqq
\left\{\left(x_i, x^{obj}_i\right) \in \R^{n_i} \times \R
: x^{obj}_i \ge f_i^*\left(x_i\right)\right\}
\]
allows us to write the objective in the maximum dot product
subject to containment in convex set form shown in the abstract:
\[
\min_{x \in \R^{\sum n_i}, A^\top x = 0}
\sum_{i=1}^m c_i^\top x_i + f_i^*\left(x_i\right)
=
\min_{
\substack{
\left[x_1; x^{obj}_1; x_2; x^{obj}_2;\ldots;x_m; x^{obj}_m\right]
\in K_1 \times \dots \times K_m\\
A^\top x = 0}}
\begin{bmatrix}c\\ \one\end{bmatrix}^{\top}
\begin{bmatrix}x \\ x^{obj} \end{bmatrix}
\]
which is the form in the abstract
\begin{equation}
\min_{\substack{x \in K_1 \times \dots \times K_m\\ A^\top x = b}}
    c^\top x
\label{eq:ermmain}
\end{equation}
with the same $d$, $m$, and $K_i$s,
$b$ set to $0$,
and $n_i$, $A$, $x$ and $c$ adjusted for the increase in row counts caused by the extra variables $x^{obj}$.

\subsection{ERM and Linear Programming}

ERM is a direct generalization of linear programming:
when $K_i = \{x : x \ge 0\}$,
\eqref{eq:ermmain} exactly reduces to the standard primal form of linear programming.
The more general form of $\ell_p$ regression, i.e., $\min_x \|Ax-b\|_p$,
is also captured by \eqref{eq:erm1} when all $n_i = 1$ and $f_i(x) = |x|^p$.

There has been a significant body of work on designing faster linear programming/GLM algorithms, largely based on interior point methods (IPMs) \cite{Vaidya89} and other second order methods. Classical IPMs for linear programming use about $\sqrt{n}$ iterations \cite{Ren88}, each of which requires solving a linear system of the form $A^\top DA$ for nonnegative diagonal matrix $D$. The recent runtime improvements largely focus on using dynamic data structures to efficient implement each iteration, sometimes in sublinear time. In the case where $n \approx d$, the state of the art runtimes for linear programming (which use strengthenings of the $\sqrt{n}$ iteration IPM) are $n^{\max\{\omega, 2+1/18\}}$ \cite{JSWZ21} where $\omega$ is the matrix multiplication exponent. There is also a corresponding result for ERM, solving \eqref{eq:ermmain} in time $n^{\max\{\omega, 2+1/6\}}$ \cite{LSZ19}, building off \cite{CLS19} who achieved the same runtime for linear programming, and recent work on $\ell_p$ regression in this regime \cite{AKPS19,AJK25}. See also \cite{Brand20} for deterministic versions of these algorithms, and \cite{Brand21} for a simplified presentation.

A somewhat separate line of work on linear programming focuses on the case where $n$ is much larger than $d$, which we refer to as the setting where the input matrix is \emph{tall}. Previous works in this setting have used an IPM of Lee-Sidford \cite{LS14} which only uses $\sqrt{d}$ iterations as opposed to $\sqrt{n}$. This opened the door to further speedups \cite{LS15,BLSS20}, and the current best runtimes are $\O(nd + d^{2.5})$ \cite{BLLSSSW21}. In other words, if the constraint matrix $A$ is sufficiently \emph{tall} ($n \ge d^{1.5}$) and is \emph{dense}, i.e., the number of nonzero entries in $A \in \mathbb{R}^{n \times d}$ is $\widetilde{\Omega}(nd)$, then the algorithm runs in nearly-linear time in the input size. However, these results have not been extended to the ERM setting, largely because it is not known whether the Lee-Sidford IPM can be extended beyond the setting of linear programming.

We also mention that several of the ideas in these works have been combined with graph theoretic primitives to design fast algorithms for maximum flow and minimum-cost flow \cite{LS14,BLNPSSSW20,BLLSSSW21,GLP21:journal,BGJLLPS22,BZ23} -- we refer the reader to \cite{CKLPPS25} for a more complete history.

\subsection{Our Results}

Our main result is an algorithm for solving the ERM problem in \eqref{eq:ermmain} in nearly-linear time for tall-dense inputs, when the dimensions $n_i$ of the underlying convex sets $K_i$ are constant (this is the same assumption as was made in \cite{LSZ19}). Formally, we assume that the sets $K_i$ are given by self-concordant barriers on them (\cref{def:sc}), and the algorithm is given access to values, gradients, and Hessians of the barrier functions at any point in constant time (we elaborate in \cref{subsec:ipmoverview}).

\begin{theorem}
\label{thm:main}
There is an algorithm that takes an ERM instance as in \eqref{eq:ermmain} such that:
\begin{enumerate}
\item each $K_i$ is given by self-concordant barriers,
bounded by $\kappa$ in magnitude $K_i \subseteq [-\kappa, \kappa]^{n_i}$,
and $n_i \leq C_K$ for some absolute constant $C_K$,
\item $A, b, c$ have entries at most $\kappa$,
and $A$ has minimum singular value at least $1/\kappa$
($A^\tp A \succeq \kappa^{-1} I$),
\end{enumerate}
outputs $x$ such that $x \in K_1 \times \dots \times K_m$, $A^\top x = b$, and
\[ c^\top x \le \eps + \min_{\substack{x \in K_1 \times \dots \times K_m\\ A^\top x = b}}
    c^\top x \]
in total time $\O(nd + d^6\sqrt{n})$,
where $\O$ hides factors of $C_K$, as well as logs of $n$, $d$, $\kappa$, and $1/\eps$.
\end{theorem}

The two assumptions we make in the statement of \cref{thm:main} are standard -- see e.g. \cite[Theorem C.3]{LSZ19}.
We remark that in the case that $n_i$ is not bounded by an absolute constant, our running time depends polynomially on $\max_{i\in[m]} n_i$.

Perhaps surprisingly, our algorithm is \emph{not} based on extending the Lee-Sidford IPM to ERM instances -- this remains an interesting open problem. Instead we argue that a $\sqrt{n}$ iteration IPM for ERM (i.e., combining the IPMs of \cite{LSZ19} and the log-barrier IPM of \cite{BLNPSSSW20}) can be implemented in nearly-linear time for tall dense instances. Our key technical advance is an algorithm that dynamically maintains a spectral sparsifier of a matrix undergoing adaptive row insertions/deletions. More precisely, the algorithm maintains leverage score overestimates of the rows of $A$ which sum to at most $\O(d)$ at all times.

\begin{theorem}
\label{thm:sparsify}
There is an algorithm that given a dynamic matrix $A \in \R^{n \times d}$ undergoing $Q$ batches of adaptive row insertions/deletions with total size at most $O(n)$,
and a parameter $\kappa$ such that at all times the Gram matrix
$A^{\tp} A$ satisfies $\frac{1}{\kappa} I \preceq A^{\tp} A \preceq \kappa I$,
maintains leverage score overestimates $\tautil_i$ for all the rows satisfying:
\begin{itemize}
\item $\tautil_i \ge a_i^\top(A^\top A)^{-1} a_i$, and
\item $\sum_{i=1}^n \widetilde{\tau}_i \le \O(d)$.
\end{itemize}
The total runtime is at most $\O(nd + Qd^{6})$.
Here (and throughout this paper)
$\O(\cdot)$ hides polylog factors in $n$, $d$, and $\kappa$.
\end{theorem}
Here, a batch means that several row insertions/deletions are all given to the algorithm at once, and all must be processed before the next batch is given.
We remark that obtaining a runtime like $\O(nd + \poly(d))$ is not possible for \cref{thm:sparsify}, which partially justifies the necessity of the additive $Qd^6$ term. Indeed, consider the case where $Q = n/d$ and each batch simply removes and adds $2d$ fresh rows, for which we have to provide leverage score overestimates. Since each instance is completely unrelated, and the best known runtime for each $2d \times d$ matrix is $\O(d^\omega)$, this input requires time at least $Qd^{\omega} = nd^{\omega-1}$, which is not bounded by $\O(nd + \poly(d))$.

\section{Preliminaries}
\label{sec:prelim}

Here we introduce the formal notations that we use throughout this paper.

\subsection{General Notation}
\label{subsec:notation}

We let $[n] \coloneqq \{1, 2, \dots, n\}$.
We use the subscript $i$ to index into functions.
So $i \in [m]$, where recall $m$ is the number of functions.
Let $S_i \subseteq [n]$ denote the set of coordinates which interact with the convex set $K_i$. Given a vector $x \in \R^n$ we let $x_i \in \R^{S_i}$ to denote the restriction of $x$ to $S_i$.

\subsection{Approximations}
\label{subsec:approx}

We use asymptotic notation and write $a \lesssim b$ as shorthand for $a = O(b)$. We write $a \approx_{\alpha} b$ if $e^{-\alpha}a \le b \le e^{\alpha} a$.

For matrices, we use Loewner ordering $A \preceq B$ to indicate $B - A$ is positive semidefinite.

We also generalize the approximation notation and use
$A \approx_{\alpha} B$ to denote $e^{-\alpha} A \preceq B \preceq e^{\alpha} A$.

\subsection{Random Projections and Heavy Hitters}
\label{subsec:hh}

This paper, like previous works on implementing IPMs with dynamic data structures, makes heavy use of $\ell_2$ sketches and heavy hitters. We start by introducing to classical JL sketch.
\begin{lemma} [Johnson–Lindenstrauss, \cite{JL84}] \label{lem:JL}
    For any $\epsjl \in (0, 1/2)$ and $n$ vectors $v_1, v_2, \dots, v_n$, let $A \sim \mathcal{N}(0,1)^{m \times d}$ where $m = O(\log n / \epsjl^2)$, it holds with probability $1 - n^{-c}$ such that 
    \[
        \left(1-\epsjl\right) \left\| v_i\right\|_2
        \le
        m^{-1/2}\left\| A v_i \right\|_2
        \le
        \left(1+\epsjl\right) \left\| v_i\right\|_2 
    \]
    for all $i \in [n]$.
\end{lemma}

We require the following standard $\ell_2$ heavy hitter data structure from \cite{KNPW11}. Please see the statement of \cite[Lemma 5.1]{GLP21:journal} for the precise statement of the Lemma below.

\begin{theorem}
\label{thm:hh}
There is an algorithm $\textsc{Build}$ that for any
error parameter $0 < \epshh < 1 / \log{n}$ and integer $n$,
$\textsc{Build}(\epshh, n)$ returns in time $\O(n)$ a random matrix $Q \in \{-1, 0, 1\}^{N \times n}$
with $N = O(\epshh^{-2}\log^3 n)$ such that every column of
$Q$ has $O(\log^3 n)$ nonzero entries.

Additionally, there is an algorithm $\textsc{Recover}$ such that for any vector
$x \in \R^n$ with $\|x\|_2 \le 1$ and access to $y = Qx \in \R^N$,
$\textsc{Recover}(y)$ returns in time $O(\epshh^{-2} \log^3 n)$
a set $S \subseteq [n]$ with size at most $O(\epshh^{-2})$
that with high probability contains all indices $i$ with $|x_i| \geq \epshh$.
\end{theorem}

The heavy hitter in \cref{thm:hh} can be used to build a data structure that 
supports updating rows of matrix $A$ and querying which rows have large norms
with respect to a given quadratic form.
We encapsulate this in the lemma below, which we prove in \cref{sec:deferred}.
\begin{lemma}\label{lem:heavy-hitter}
There exists a randomized data structure \textsc{HeavyHitter} that maintains a set of vectors $a_1 \ldots a_n \in \R^{d}$
under the following operations against a non-adaptive adversary:
\begin{itemize}
\item Initialize in time $\Otil(n d)$.
\item $\textsc{Modify}(i, v)$: 
Set $a_i \leftarrow v$
in time $\Otil(d)$, where $v = 0$ is equivalent to deleting it.
\item $\textsc{Query}(M, \delta)$
Given a polynomially-conditioned symmetric PSD matrix $M \in \R^{d \times d}$,
and a threshold $\delta > 0$,
return a set of $O(\delta^{-1} \sum_{1 \leq i \leq n} \| a_i\|_{M}^{2})$ indices that include all $i$ such that
\[
\norm{ a_i }_{M}^2 \geq \delta
\]
in time $\Otil(d^{\omega}
+ \delta^{-1} d \sum_{1 \leq i \leq n} \| a_i\|_{M}^{2})$.
That is, $\O(d)$ times the maximum number of rows which may exceed the threshold, plus matrix multiplication time.
\end{itemize}
\end{lemma}

\subsection{Leverage Score Sampling}

We require the following standard lemma which says that sampling by leverage score overestimates produces a spectral sparsifier with high probability. We use a slightly adapted version where the leverage scores and sparsifier error are computed with respect to a different matrix $M$.
\begin{lemma}
\label{lem:sample}
If $A = [a_1^{\top}, a_2^{\top}, \ldots, a_n^{\top}]$ is a $n \times d$ matrix,
$M$ is a $d \times d$ symmetric positive definite matrix,
and $w_i$s are values such that
\[
w_i \geq a_i^{\top} M^{-1} a_i.
\]
Let $p_i = \frac{w_i}{\sum_{i=1}^n w_i}$ and for $j = 1, \dots, T \coloneqq 100\eps^{-2}\log n \cdot \sum_{i=1}^n w_i$ let $i_j$ be a random $i \in [n]$ selected with probability $p_i$. Let $\Atil \in \R^{T \times d}$ be a matrix whose $j$-th row is $(p_{i_j}T)^{-1/2} a_{i_j}$.
Then whp:
\[
-\epsilon M \preceq A^\top A - \Atil^{\top} \Atil \preceq \epsilon M.
\]
\end{lemma}

\section{Overview}
\label{sec:overview}

\subsection{IPM Setup}
\label{subsec:ipmoverview}

Towards formally setting up the statement and proof of \cref{thm:main}, we need to define our access model to the convex sets $K$. Here, following \cite{LSZ19} we assume that the algorithm has access to a self-concordant (SC) barrier on each $K_i$. It is known that every convex set $K_i \subseteq \R^{n_i}$ admits a $\nu_i \le n_i$ self-concordant barrier \cite{NN94,LY21,Chewi21}, and recall $n_i$ is upper bounded by an absolute constant $C_k$. Additionally, most functions admit simple to express $O(\nu_i)$-SC barriers \cite{NN94}.
Thus we assume that the algorithm has access to evaluation, gradient, and Hessian oracles to $\nu_i$-SC barriers on each set $K_i$, where each oracle call takes $O(1)$ time.

\begin{definition}[Self-concordance]
\label{def:sc}
For a convex set $K \subseteq \R^n$ we say that a convex function $\phi: \inte(K) \to \R$ is $\nu$-self-concordant if:
\begin{enumerate}
    \item (Self-concordance) For all $x \in \inte(K)$ and $u, v, w \in \R^n$ it holds that
    \[ \left|\g^3 \phi(x)[u, v, w] \right| \le 2\left(u^\top \g^2 \phi(x) u\right)^{1/2}\left(v^\top \g^2 \phi(x) v\right)^{1/2}\left(w^\top \g^2 \phi(x) w\right)^{1/2}, \enspace \text{ and } \enspace \]
    \item For all $x \in \inte(K)$ it holds that $(\g \phi(x))^\top \g^2 \phi(x)^{-1}(\g \phi(x)) \le \nu$.
\end{enumerate}
\end{definition}
Informally, the first property says that if $x$ does not move too much (measured in the norm induced by the local Hessian at $x$), then the quadratic form of the Hessian also does not change much spectrally. A standard IPM tracks a \emph{central path} of points defined using the self-concordant barrier functions. More precisely, for a parameter $t > 0$, define
\begin{equation}\label{eq:path}  x^{(t)} = \argmin_{A^\top x = b} c^\top x + t \sum_{i=1}^m \phi_i(x_i), \end{equation}
where $x_i$ is the restriction of $x$ to the coordinates corresponding to $K_i$. The KKT conditions for this can be expressed as
\[ c + t \g \Phi(x) = Ay \enspace \text{ for some } \enspace y \in \R^d, \] where $\Phi(x) = \sum_{i=1}^m \phi_i(x_i)$. This can equivalently be written as $s_i/t + \g \phi_i(x_i) = 0$ for all $i = 1, \dots, m$ where $s = c - Ay$ are the slacks. In this way, we call a pair of $x, s$ satisfying these properties for $t$ a \emph{well-centered} pair (see \cref{def:wellcenter} for a more precise definition).

\subsection{Overview of Robust IPM for ERM}
\label{subsec:overviewerm}

The goal of an IPM is to ``follow the central path'', i.e., slowly decrease $t$ towards $0$ while maintaining $(x, s)$ that are well-centered for that value of $t$. For the purposes of being able to implement the IPM efficiently, we work with a very loose notion of centrality introduced in \cite{CLS19}, where we only assert that $(x, s)$ is within some $\ell_\infty$ ball of the central path, as opposed to an $\ell_2$ ball (which is more standard). This is formally captured by the exponential/softmax potential function defined in \eqref{eq:psi}.

When taking a step to update $(x, s)$ the algorithm needs to solve a linear system in the matrix $A^\top \g^2 \Phi(x) A$. Here, note that $\g^2 \Phi(x)$ is a block-diagonal matrix with block sizes $n_i \times n_i$. However, just as is done in previous works on nearly-linear time linear programming \cite{BLSS20,BLNPSSSW20,BLLSSSW21}, the algorithm instead computes a spectral sparsifier of this matrix to use instead. The spectral sparsifier is sampled using leverage score overestimates, which explains why we need our new data structure for dynamic leverage score maintenance in \cref{thm:sparsify}.

In each step, we also need to maintain approximations $\xbar, \sbar$ to $x, s$ that are used instead of $x, s$ to define the step. These approximations again are with respect to $\ell_\infty$. It can be proven that there exists such $\xbar$ and $\sbar$ so that only $\O(n)$ total coordinates in $\xbar$ and $\sbar$ change over the course of the whole algorithm -- this is a standard fact from IPM stability analysis. Algorithmically maintaining $\sbar$ requires heavy-hitter data structures which have already been well-developed in the linear programming setting, and simple modifications extend it to the ERM setting without much challenge. 
Maintaining $\xbar$ is a bit trickier, but has also been worked out in the linear programming setting (see eg. \cite{BLNPSSSW20,BLLSSSW21}). The idea is that the change in $x$ can be subsampled down to support size about $\O(\sqrt{n} + d)$ (plus a gradient term which is easy to maintain) instead of the total size $n$. This allows us to both maintain $\xbar$ cheaply as well as the ``feasibility error'' $A^\top x - b$ resulting from the use of the sparsifier.

\subsection{Overview of Adaptive Sparsifier Algorithm}

In this section we overview the algorithm for \cref{thm:sparsify}, i.e., dynamic leverage score overestimate maintenance against an adaptive adversary. As described in \cref{subsec:overviewerm}, we will sample by these leverage score overestimates to produce a spectral sparsifier to use within the IPM.

\paragraph{Decremental sparsifier.}
As is now standard in the dynamic algorithms literature, a fully dynamic data structure follows fairly easily from a decremental one (i.e., one that only undergoes row deletions / downscalings), and we briefly describe this reduction at the end. At a high level, given a matrix $A \in \R^{n \times d}$ undergoing row halvings, our goal is to detect anytime that the leverage score of a row increased additively by more than $d / n$.

Our decremental data structure is from combining the following two facts used in several previous works on electrical flows~\cite{CKMST11} and online sparsification~\cite{CMP20}: 
\begin{enumerate}
\item If we remove (fractional) rows from $A$ whose total leverage score is at most $0.5$, then no leverage score of the remaining rows has more than doubled. This allows us to wait until enough changes have accumulated before having to update the leverage scores.
\item Deleting a row with leverage score $\tau$ decreases the determinant of $A^\top A$ by a factor of $(1-\tau)$. Thus if $A$ has polynomially lower and upper bounded singular values at all times, the total multiplicative decrease of $\det(A^\top A)$ is at most $n^{O(d)}$, so the sum of leverage scores of deleted rows is at most $O(d \log n)$. 
\end{enumerate}
The second fact, combined with the total sum of leverage scores is at most $d$,
implies that the total increase in leverage scores across all steps is $O(d \log{n})$.
In other words, only $\O(n)$ additive changes of leverage scores by $d / n$ will be detected.
Furthermore, combining these two facts gives that the
number of phases where we go and look for new leverage score estimates is $\O(d)$:
this much lower number of phases (compared to the $n^{1/2}$ iterations of the IPM) is critical to setting errors in heavy-hitter sketches.

\paragraph{Detecting large leverage score changes.} To implement the algorithm described above we need to detect when rows' leverage scores have increased. For this we use a heavy-hitter data structure (see \cref{thm:hh}) along with a standard dyadic interval trick. We defer the details to \cref{sec:sparsifier}.

In this overview we instead discuss how we handle the issue of adaptive adversaries in the data structure.
For this we use a locator/checker framework which has been used in several past dynamic leverage score maintenance data structures \cite{FMPSWX18,GLP21:journal,BGJLLPS22}. The goal of the locator is to detect a set $S$ of edges on which to check the leverage scores: this set is guaranteed to contain any edge whose leverage score we ultimately update. This is where the heavy hitter data structure is used. The checker takes all the edges in $S$ and estimates their leverage scores to decide which ones have large leverage score -- for this sampling a spectral sparsifier and using a Johnson-Lindenstrauss sketch suffices. The checker is resampled at each iteration to be a fresh spectral sparsifier. This way, the randomness between the locator and checker is independent, and no randomness of the locator (besides very low probability events) leaks between iterations.

\paragraph{From decremental halving to fully-dynamic.}

One issue that arises is that rows may have leverage score $1$:
deleting a row no longer leaves us with leverage score approximations.
To handle this, we instead only halve rows,
or equivalently, delete rows fractionally.
This increases the number of operations by $O(\log{n})$,
but ensures that the outer-product of $A$,
and in particular, all leverage scores,
are preserved multiplicatively across each step.
It in turn allows us to use previous leverage scores
to sample the current matrix, only paying a constant factor
increase in the number of row samples in $\Atil$.

A fully dynamic algorithm may have insertions. To obtain this, we maintain an $O(\log n)$ level data structure, where the $k$-th level from the bottom handles the most recent $2^k$ insertions. Every $2^k$ insertions, we clear the bottom $k$ levels and rebuild them using any of the $2^k$ insertions which haven't been deleted yet.

\subsection{Overall Runtime Analysis}

To implement the algorithm described, we need to discuss how to maintain $\xbar$, $\sbar$, the feasibility error $A^\top x - b$, and the sparsifier we require at each iteration. In short, these are handled as follows, and is mostly based on prior works \cite{BLSS20,BLNPSSSW20,BLLSSSW21}.
\begin{itemize}
    \item (Feasibility maintenance): At a high level, the change to $x$ during each iteration takes the following form: $x \to x - (g - R\delta)$, where $g$ is the gradient (a slowly changing vector itself), and $R$ is a $\O(\sqrt{n}+d)$-sparse diagonal matrix, so that $R\delta$ is a vector of sparsity at most $\O(\sqrt{n} + d)$. Thus, $A^\top x - b$ can be maintained by calculating $A^\top R\delta$ explicitly in time $O(d \cdot (\sqrt{n}+d))$ (which is acceptable), and then maintaining a partial-sum data structure to maintain $A^\top g$.
    
    \item ($\xbar$ maintenance): We prove that even with the subsampling procedure that the changes to $x$ are large only at most $\O(n)$ times throughout the algorithm. These changes can be detected mostly explicitly: track the changes to $g$ and the other coordinate changes to $x$ explicitly. This is formally done by arguing that there is a nearby sequence $\hat{x}$ that is $\ell_2$-stable (see \cref{lem:xstable}).

    \item ($\sbar$ maintenance): This is done by using a heavy hitter data structure. Because the update structure of $s$ is $s \to DAx$ each iteration for a slowly changing diagonal matrix $D$, we can use an $\ell_2$ heavy hitter data structure to detect large changes to $s$ (see \cref{lem:slack}).

    \item (Subsampling changes in $x$): This is done by sampling by using a combination of leverage scores and a heavy-hitter/JL data structure (see \cref{lem:validsampling}).

    \item (Sparsifier): The time cost of the sparsifier is dominated by \cref{thm:sparsify}, which costs $\O(nd + d^6\sqrt{n})$ because we have $\O(\sqrt{n})$ batches (one per iteration of the IPM), and up to $\O(n)$ total row updates.
\end{itemize}

In total, the sparsifier dominates the cost of the IPM and costs time $\O(nd + d^6\sqrt{n})$.

\section{Single-Step Robust IPM for ERM}
\label{sec:ipm}

In this section we give the algorithm which takes one step along the central path and analyzes that step. We start by formally introducing the self-concordant barrier functions that describe the convex sets $K_i$ and other useful notation.

\subsection{Formal setup}

We assume that the algorithm is given access to all higher-order derivatives of self-concordant barrier functions $\phi_i: \mathsf{int}(K_i) \to \R$. We assume that $\phi_i$ is $\nu_i$-self-concordant for constants $\nu_i$.

The coordinates of $x$ which interact with $K_i$ are a subset of $[n]$ of size $n_i$: we call these coordinates a \emph{block}. The $i$-th block is the set of coordinates for the set $K_i$, and for any vector $v \in \R^n$ we let $v_i \in \R^{n_i}$ be the restriction of $v$ to the $i$-th block.

We also give notation to express the maximum of $\ell_2$ norms over blocks.
\begin{definition}
    For $w \ge 1$ and a vector $v \in \R^n$, define the $\|v\|_{\infty,w} \coloneqq \max_{i\in[m]}\|v_i\|_w$, where $v_i$ denote the restriction of $v$ to the $i$-th block.
\end{definition}

\subsection{Potential Function Setup}

To follow the central path, consider the optimality conditions of \eqref{eq:path}. Recall $s=c-Ay$ are the slacks, and for optimality we need $s_i/t+\nabla\phi_i(x_i)=0$ for all $i\in[m]$. Note that unlike linear programs our slacks $s$ can be negative.

Accordingly, we define the centrality error vector for a slack/primal pair as
\begin{equation} \mu^t_i(x, s) \coloneqq \frac{s_i}{t} + \g \phi_i(x_i) \enspace \text{ for } i \in [m]. \label{eq:mu} \end{equation}
Now we define the centrality error for a block $i \in [m]$
as the norm of the centrality error vector in the inverse Hessian norm, i.e.,
\begin{equation}
\gamma^t_i(x, s) \coloneqq \|\mu^t_i(x, s)\|_{\g^2 \phi_i(x_i)^{-1}}^2. \label{eq:gamma}
\end{equation}

Let $\eps<1/80$ be fixed and let $\lambda = \frac{C_{center} \log n}{\eps^2}$.
Then the centrality potential is defined as
\begin{equation}
\Psi^t(x, s) \coloneqq \sum_{i=1}^m \exp(\lambda\gamma^t_i(x,s)). \label{eq:psi}
\end{equation}

Now we define the \emph{feasibility error} of $x$ as
\[
\left\|A^\top x - b\right\|_{\left(A^\top \g^2\Phi\left(x\right)^{-1} A\right)^{-1}}.
\]
This error is part of the centering condition in \Cref{def:wellcenter}.
We correct for it by taking steps in the direction of its gradient,
and control it in \Cref{lem:feasible} by showing that adequate steps
can decrease it quadratically.

Together these let us define a well-centered pair $(x, s)$ at a path parameter $t$.
\begin{definition}
\label{def:wellcenter}
We say that a pair $(x, s)$ is $\eps$ well-centered at a path parameter $t$ if:
\begin{enumerate}
    \item (Centrality) $\gamma^t_i(x, s) \le \eps^2$ for all $i \in [m]$, and
    \item (Primal Feasibility) $\|A^\top x - b\|_{\left(A^\top \g^2\Phi(x)^{-1} A\right)^{-1}} \le \alpha\eps$, and
    \item (Dual Feasibility) $s = c - Ay$ for some $y \in \R^d$.
\end{enumerate}
\end{definition}

Next we define the steps we take to decrease the centrality potential defined in \eqref{eq:psi}. Towards this we define the \emph{gradient} and the \emph{ideal step}. Ultimately our algorithm will take the ideal step defined for approximate $x, s$ and a sparsifier of the true Hessian.

As standard to the robust IPM literature, we use $g^{t}(x, s) \in \R^{m}$
to denote the ideal change that we want a change in $x$ and $s$ to send
each of the $\gamma^{t}_i$s in.

\begin{definition}[Gradient]
\label{def:gradient}
Given $x, s, t$, we define the gradient $g^t(x, s) \in \R^n$ as
\[
g_i^t\left(x, s\right)
\coloneqq
\frac{\exp\left(\lambda \gamma_i^t\left(x,s\right)\right)
\cdot \g^2\phi\left(x_i\right)^{-1/2}\mu_i^t\left(x,s\right)}
{\left(\sum_{i=1}^m \exp\left(2\lambda\gamma_i^t\left(x,s\right)\right)\right)^{1/2}}
\qquad
\text{for $i \in [m]$}.
\]
\end{definition}

Finally, for our algorithm we do not use the exact values of $x$ or $s$ and instead internally maintain approximations $\xbar$ and $\sbar$ for them. We need to define what it means for $\xbar$ and $\sbar$ to $\eps$-approximate the true $x$ and $s$ values.
\begin{definition}
\label{def:epsapprox}
We say that $\xbar$ and $\sbar$ $\eps$-approximate $x, s$ if
\[
\norm{x_i - \xbar_i}_{\g^2 \phi_i\left(x_i\right)}
\le
\eps
\enspace \text{ and } \enspace
\norm{s_i - \sbar_i}_{\g^2\phi_i\left(x_i\right)^{-1}}
\le
\eps t.
\]
\end{definition}

\begin{definition}[Ideal step]
\label{def:idealstep}
Given $x, s, t$ define the \emph{ideal step} for $g = g^t(x, s)$ as
\begin{align*}
\delta_x &= \g^2\Phi(x)^{-1/2}g - \g^2\Phi(x)^{-1}A(A^\top \g^2\Phi(x)^{-1}A)^{-1}A^\top \g^2\Phi(x)^{-1/2}g \enspace \text{ and } \\
\delta_s &= t \cdot A(A^\top \g^2\Phi(x)^{-1}A)^{-1}A^\top \g^2\Phi(x)^{-1/2}g.
\end{align*}
\end{definition}

\subsection{Short-Step Analysis}

To start we state the short step algorithm. This is based on previous works \cite{BLSS20,BLNPSSSW20,BLLSSSW21} but adapted to the ERM setting using the setup of \cite{LSZ19}.
\begin{algorithm}
\SetKwProg{myalg}{Procedure}{}{}
\myalg{$\textsc{ShortStep}(x, s, t, \eta)$}{
    \tcp{Let $\lambda\assign C_{center} \eps^{-2} \log n$,
    $\alpha\gets \eps C_K^{-1} \lambda^{-1}$,
    $\beta \gets 10\alpha$
    } 
    Let $\xbar, \sbar$ be $\beta$-approximations of $x, s$ (see \cref{def:epsapprox}). \\
    Let $g = \alpha g^t(\xbar, \sbar)$ (see \cref{def:gradient}). \\
    Let $H \approx_{\alpha} A^\top \g^2 \Phi(\xbar)^{-1} A.$ \\
    Let
    \begin{align*}
    \delta_1 &= \g^2 \Phi(\xbar)^{-1/2}AH^{-1}A^\top\nabla^2\Phi(\xbar)^{-1/2} g\\
    \delta_2 &= \g^2\Phi(\xbar)^{-1/2}AH^{-1}(A^\top x - b)\\
    \delta_r &= \delta_1 + \delta_2
    \end{align*}
    Let $R$ be a \emph{valid} diagonal matrix sample for vector $\delta_r$ and matrix $\g^2\Phi(\xbar)^{-1/2}A$ (see \cref{def:valid}). \\
    Set $\bar{\delta}_x = \g^2\Phi(\xbar)^{-1/2}(g - R\delta_r)$ and $\bar{\delta}_s = t\g^2\Phi(\xbar)^{1/2}\delta_1$. \\
    \Return $x^{\new} = x - \bar{\delta}_x$ and $s^{\new} = s - \bar{\delta}_s$.
}
\caption{Short Step IPM for ERM: starting at $x, s$ at path parameter $t$, decrease $t$ to $(1-\eta)t$ and update $x$ and $s$ to $x^{\new}$ and $s^{\new}$.}
\label{alg:shortstep}
\end{algorithm}

\begin{definition}
\label{def:valid}
We say that a random nonnegative diagonal matrix $R \in \R_{\ge0}^{n \times n}$ is a \emph{valid sample} for a vector $\delta$ and matrix $A$ if for a sufficiently large constant $C_{var}$:
\begin{enumerate}
    \item (Block form) For coordinates $i$ and $j$ in the same block, $R_{ii} = R_{jj}$, and
    \item (Expectation) It holds that $\E[R] = I$, and
    \item (Variance) It holds that $\Var[R_{ii}\delta_i] \le \frac{\alpha|\delta_i|\|\delta\|_2}{C_{var}^2}$, and
    \item (Covariance) For coordinates $i$ and $j$ in different blocks, it holds that $\E[R_{ii}R_{jj}] \le 2$, and
    \item \label{item:valid_maximum} (Maximum) With high probability, it holds that $\|R\delta - \delta\|_{\infty} \le \frac{\alpha\|\delta\|_2}{C_{var}^2}$, and
    \item (Spectral approximation) With high probability, it holds that
    \[ A^\top \g^2 \Phi(x)^{-1/2} R \g^2 \Phi(x)^{-1/2}A \approx_{\alpha} A^\top \Phi(x)^{-1} A. \]
\end{enumerate}
\end{definition}

Next we state the main lemmas which prove that the short-step procedure in \cref{alg:shortstep} indeed maintains a sequence of well-centered points. Later, we argue that $\xbar$ and $\sbar$ change slowly over a sequence of short-steps, and give efficient algorithms for maintaining them.

\begin{lemma}[Potential Maintainence]
\label{lem:pot_drop}
Assume that $(x, s)$ are $\eps$ well-centered at path parameter $t$.
Let $\that = (1-\eta)t$ for $\eta = \frac{\eps\alpha}{C_{center} \sqrt{\nu}}$. It holds that
\[ \E[\Psi^{\that}(x^{\new}, s^{\new})] \le \left(1 - \frac{\eps\alpha}{C_{center} ^2\sqrt{\nu}}\right)\Psi^t(x, s)+ O(n^2). \]
\end{lemma}

Before proving this lemma, we first analyze the change in potential and establish some bounds on the size of the steps we take. Observe the following technical lemma.

\begin{lemma}\label{lem:exp_analysis}
    \[ \exp(\lambda(\gamma+\delta_{\gamma}))\le\exp(\lambda\gamma)(1+\lambda\delta_{\gamma}+\exp(\lambda|\gamma|)\lambda^2\delta_{\gamma}^2)).\]
\end{lemma}
\begin{proof}
    Consider $t\in[0,1]$ and let 
    \[
    z_t \coloneqq \gamma+t\delta_{\gamma},
    \]
    and let
    \[
    f\left(t\right)\coloneqq \exp\left(\lambda z_t\right).
    \]
    Taylor's theorem tells us
    \[
    f\left(1\right)
    =
    f\left(0\right)
    +f'\left(0\right)
    +\frac12f''\left(\zeta\right)
    \]
    for $\zeta\in[0,1]$. We bound these terms separately. 
    
    The first term is simply $\exp(\gamma)$.
    
    By the chain rule,
    \[
    f'\left(t\right)
    =
    \exp\left(\lambda z_t\right)
    \lambda\frac{d}{dt}z_t
    =
    \exp\left(\lambda z_t\right)\lambda\delta_{\gamma}.
    \]
    For $t=0$ this is $\exp(\lambda\gamma)\lambda\delta_\gamma$.
    
    Again by the chain rule $f''(t)=\exp(\lambda z_t)\lambda^2\delta_{\gamma}^2$,
    which for some $t=\zeta$ is
    \[
    \exp\left(\lambda\gamma\right)
    \exp\left(\lambda\zeta\delta_\gamma\right)\lambda^2\delta_\gamma^2.
    \]
    Since $\zeta\in[0,1]$, this is upper bounded by $\exp(\lambda\gamma)\exp(\lambda|\delta_\gamma|)\lambda^2\delta_{\gamma}^2)$. Summing the terms gives us the desired result.
\end{proof}

\begin{lemma}\label{lem:gamma_approx}
Let $\xbar$ and $\sbar$ be $\beta$-approximations to $x$ and $s$, under \Cref{def:epsapprox}. Then
\[
\abs{\gamma_i^t\left(x,s\right)
-
\gamma_i^t\left(\xbar,\sbar\right)}
\le
10\beta\eps.
\]
\end{lemma}

\begin{proof}
We bound using self-concordance
\begin{align*}
\abs{\gamma_i^t\left(x,s\right) - \gamma_i^t\left(\xbar,\sbar\right)}
&=
\abs{\norm{\mu_i^t\left(x,s\right)}_{\nabla^2\phi_i\left(x_i\right)^{-1}}^2-
\norm{\mu_i^t(\xbar,\sbar)}_{\nabla^2\phi_x\left(\xbar_i\right)^{-1}}^2}\\
&\le\abs{\norm{\mu_i^t\left(x,s\right)}_{\nabla^2\phi_i(x_i)^{-1}}^2
- \norm{\mu_i^t\left(x,s\right)}_{\nabla^2\phi_i\left(\xbar_i\right)^{-1}}^2}
\\ & \qquad
+\abs{\norm{\mu_i^t\left(x,s\right)}_{\nabla^2\phi_i\left(\xbar_i\right)^{-1}}^2
-\norm{\mu_i^t\left(\xbar,\sbar\right)}_{\nabla^2\phi_x\left(\xbar_i\right)^{-1}}^2}\\
&\le
\left(\frac{1}{1-\beta}-1\right)\gamma_i^t\left(x,s\right)
+3\eps \norm{\mu_i^t\left(x,s\right)-\mu_i^t\left(\xbar,\sbar\right)}
_{\nabla^2\phi_x\left(\xbar_i\right)^{-1}}\\
&\le
\left(\frac{1}{1-\beta}-1\right)\gamma_i^t\left(x,s\right)
+\frac{3\eps}{1-\beta}\norm{\frac{\sbar_i-s_i}{t}+\nabla\phi_i\left(\xbar_i\right)-\nabla\phi_i\left(x_i\right)}
_{\nabla^2\phi_x\left(x_i\right)^{-1}}\\
&\le
\left(\frac{1}{1-\beta}-1\right)\eps^2+\frac{3\eps}{1-\beta}\left(\beta+\frac{\beta}{1-\beta}\right),
\end{align*}
where in the third step we pulled out a $\norm{\mu_i^t\left(x,s\right)}_{\nabla^2\phi_i\left(\xbar_i\right)^{-1}}
+\norm{\mu_i^t\left(\xbar,\sbar\right)}_{\nabla^2\phi_x\left(\xbar_i\right)^{-1}}$ using difference of squares and Cauchy-Schwarz.
For $\beta<\eps<0.1$ this is bounded by $10\beta\eps$ as desired.
\end{proof}

\begin{lemma}\label{lem:potential}
    Let $\Psi(\mu)=\exp(\lambda\|\mu\|_{M_\mu}^2)$, $\mu^\new=\mu-\delta_\mu$ for $\|\delta_\mu\|_{M}\le\eps^2$, and $M_{\mu^\new}\approx_{\eps^2}M_\mu$. Then
    \[ \Psi(\mu^\new)\le\Psi(\mu)-2\exp(\lambda\|\mu\|_{M_\mu}^2)\lambda\delta_\mu^\top M_\mu\mu+4\exp(\lambda\|\mu\|_{M_\mu}^2)\exp(2\lambda\eps\|\mu\|_{M_\mu})\lambda^2\eps^2\|\mu\|_{M_\mu}^2. \]
\end{lemma}
\begin{proof}
We first consider the $\|\mu\|_{M_\mu}$ term. 
\begin{align*}
\norm{\mu^\new}_{M_\mu^\new}^2
&=
\left(\mu^\new\right)^\top M_{\mu^\new}\mu^\new\\
&=
\left(\mu-\delta_\mu\right)^\top M_{\mu^\new} \left(\mu-\delta_\mu\right)\\
&\le
\exp\left(2\eps^2\right)\left(\mu-\delta_\mu\right)^\top
M_{\mu}\left(\mu-\delta_\mu\right)\\
&=
\exp\left(2\eps^2\right)\left(\norm{\mu}_M^2-2\delta_\mu^\top M_\mu\mu
+\norm{\delta_\mu}_{M_\mu}^2\right)\\
&\le
\norm{\mu}_M^2-2\delta_\mu^\top M_\mu\mu + O\left(\eps^2\right).
\end{align*}
Then, we consider the change over $\Psi$.
Let
\[
\gamma
\coloneqq
\norm{\mu}_M^2,
\]
and
\[
\delta_\gamma
\coloneqq-2\delta_\mu^\top M_\mu\mu+O\left(\eps^2\right).
\]
Then we use \Cref{lem:exp_analysis}. To bound $\delta_\gamma$, note the following by Cauchy-Schwarz:
\[
\delta_\gamma^2
=4\left(\delta_\mu^\top M_\mu\mu\right)^2
\le
4\norm{\delta_\mu}_{M_\mu}^2 \norm{\mu}_{M_\mu}^2
\le
4\eps^4 \norm{\mu}_{M_\mu}^2.
\]
Combining these, we have
\[
\Psi\left(\mu^\new\right)
\le
\Psi\left(\mu\right)
-
2\exp\left(\lambda\norm{\mu}_{M_\mu}^2\right)\lambda\delta_\mu^\top M_\mu\mu
+
4\exp\left(\lambda\norm{\mu}_{M_\mu}^2\right)
\exp\left(2\lambda\eps^2\norm{\mu}_{M_\mu}\right)
\lambda^2\eps^4\norm{\mu}_{M_\mu}^2.
\]
\end{proof}

\begin{proof}[Proof of \Cref{lem:pot_drop}]
We want to use \Cref{lem:potential} for $\mu=\mu_i^t(x,s)$ and $M=\nabla^2\phi_i(x_i)^{-1}$. The stability of $M$ follows simply by self-concordance, so we analyze the effect of changing $x$ and $s$ on $\mu$.
Recall $x^\new=x-\bar\delta_x$, $s^\new=s-\bar\delta_s$,
and $\mu_i(x,s)=\frac{s_i}{t}+\nabla\phi_i(x_i)$.
Then, by the definition of $\mu$ and self-concordance of $\phi_i$,
we have
\begin{align*}
\mu_i^t\left(x^\new,s^\new\right)
&=
\mu_i^t\left(x,s\right) -\frac{\bar\delta_{s,i}}{t}-\nabla^2\phi_i(x_i)\bar\delta_{x,i}+\frac{\norm{\bar\delta_{x,i}}_{\nabla^2\phi_i\left(x_i\right)}^2}{1-\norm{\bar\delta_{x,i}}_{\nabla^2\phi_i\left(x_i\right)}}\\
&\le
\mu_i^t\left(x, s\right) - \frac{\bar\delta_{s,i}}{t}-\nabla^2\phi_i\left(x_i\right)\bar\delta_{x,i}+2\frac{\eps}{C_{center}^2\lambda},
\end{align*}
where the second inequality follows from \Cref{lem:stability}.

We now calculate the change:
\begin{align*}
\norm{\E\left[\frac{\bar\delta_{s,i}}{t}+\nabla^2\phi_i(x_i)\bar\delta_{x,i}\right]}
_{\nabla^2\phi_i\left(x_i\right)^{-1}}
&=
\norm{\nabla^2\phi_i\left(x_i\right)^{-1/2}
\left(\frac{\bar\delta_{s,i}}{t}+\nabla^2\phi_i\left(x_i\right)\E\left[\bar\delta_{x,i}\right]\right)}_2\\
&\le
\norm{\nabla^2\phi_i(x_i)^{-1/2}\frac{\bar\delta_{s,i}}{t}}_2
+
\norm{\nabla^2\phi_i(x_i)^{-1/2}\nabla^2\phi_i\left(x_i\right)
\E\left[\bar\delta_{x,i}\right]}_2\\
&\lesssim
\norm{\delta_1}_2
+
\norm{\nabla^2\phi_i(x_i)^{1/2}\E\left[\bar\delta_{x,i}\right]}_2\\
&\le
7\alpha\eps
\end{align*}
where the final inequality comes from \Cref{lem:step_size} and \Cref{lem:stability}.
Then, applying \Cref{lem:potential} to each block and summing, we have
\begin{align*}
\Psi^t\left(x^\new,s^\new\right)
&\le
\Psi^t\left(x,s\right)\\
&\qquad
-2\sum_{i=1}^m\exp\left(\lambda\gamma_i^t\left(x,s\right)\right)
\lambda\left(\frac{\bar\delta_{s,i}}{t}
  +\nabla^2\phi_i\left(x_i\right)\bar\delta_{x,i}\right)^\top 
\nabla^2\phi_i\left(x_i\right)^{-1}\mu_i^t\left(x,s\right)\\
&\qquad \qquad
+4\sum_{i=1}^m
\exp\left(\lambda\gamma_i^t\left(x,s\right)\right)
\exp\left(2\lambda\eps\gamma_i^t\left(x,s\right)^{1/2}\right)
\lambda^2\eps^2\gamma_i^t\left(x,s\right).
\end{align*}
As $\gamma_i^t(x,s)\le\eps^2$ by centrality of $x$ and $s$ (\Cref{def:wellcenter}),
the third term is on the scale of $\eps^4n^\eps$ and can be ignored.

We consider the first order term, without the scaling for now:
\begin{align*}
&\E\left[\left(\frac{\bar\delta_{s,i}}{t}+\nabla^2\phi_i\left(x_i\right)\bar\delta_{x,i}\right)^\top 
\nabla^2\phi_i\left(x_i\right)^{-1}\mu_i^t\left(x,s\right)\right]\\
&\qquad=
\left(\nabla^2\phi_i\left(x_i\right)^{1/2}g
-
AH^{-1}\left(A^\top x -b\right)\right)^\top
\nabla^2\phi_i\left(x_i\right)^{-1}\mu_i^t\left(x,s\right)\\
&\qquad=
g^\top \nabla^2 \phi_i\left(x_i\right)^{-1/2} \mu_i^t\left(x,s\right)-\left(\nabla^2\phi_i\left(x_i\right)^{-1/2}AH^{-1}\left(A^\top x-b\right)\right)^\top\nabla^2\phi_i\left(x_i\right)^{-1/2}\mu_i^t\left(x,s\right)\\
&\qquad\ge
\frac{\alpha\exp\left(\lambda \gamma_i^t\left(\xbar,\xbar\right)\right)}{\left(\sum_{j=1}^m \exp\left(2\lambda\gamma_j^t\left(\xbar,\sbar\right)\right)\right)^{1/2}}
\cdot
\mu_i^t\left(\xbar,\sbar\right)^\top
\nabla^2\phi_i\left(\xbar_i\right)^{-1/2}\nabla^2\phi_i\left(x_i\right)^{-1/2}
\mu_i^t\left(x,s\right)-\norm{\delta_2}_2\gamma_i^t\left(x,s\right)^{1/2}\\
&\qquad\gtrsim
\frac{\alpha\exp\left(\lambda \gamma_i^t\left(\xbar,\sbar\right)\right)}
{\left(\sum_{j=1}^m \exp\left(2\lambda\gamma_j^t\left(\xbar,\sbar\right)\right)\right)^{1/2}}
\cdot \mu_i^t\left(\xbar,\sbar\right)^\top
\nabla^2\phi_i\left(\xbar_i\right)^{-1}\mu_i^t\left(\xbar,\sbar\right)-\norm{\delta_2}_2\gamma_i^t\left(x,s\right)^{1/2}\\
&\qquad\ge
\frac{\alpha\exp\left(\lambda \gamma_i^t\left(\xbar,\sbar\right)\right)}
{\left(\sum_{j=1}^m \exp\left(2\lambda\gamma_j^t\left(\xbar,\sbar\right)\right)\right)^{1/2}}
\cdot\gamma_i^t\left(\xbar,\sbar\right) - \alpha\eps^2.
\end{align*}
Thus the new potential can be bounded by
\[ \Psi^t\left(x^\new,s^\new\right)
\le
\Psi^t\left(x,s\right)
-
2\alpha\lambda\frac{\sum_{i=1}^m
\exp\left(2\lambda \gamma_i^t\left(\xbar,\sbar\right)\right)
\cdot
\gamma_i^t\left(\xbar,\sbar\right)}
{\left(\sum_{j=1}^m \exp\left(2\lambda\gamma_j^t\left(\xbar,\sbar\right)\right)\right)^{1/2}}.
\]
Consider the numerator. Let $\eps_0=\eps^2/2C_{center}$ so that $\exp(2\lambda\eps_0)=n$,
where recall $\lambda=C_{center} \log n/\eps^2$. Then
\begin{align*} \sum_{i=1}^m\exp\left(2\lambda\gamma_i^t\left(\xbar,\sbar\right)\right)
\cdot
\gamma_i^t\left(\xbar,\sbar\right)
&\ge
\sum_{i:\gamma_i^t\left(\xbar,\sbar\right)\ge\eps_0}
\eps_0\exp\left(2\lambda\gamma_i^t\left(\xbar,\sbar\right)\right)\\
&=
\eps_0\sum_{i=1}^m\exp\left(2\lambda\gamma_i^t\left(\xbar,\sbar\right)\right)
-
\eps_0\sum_{i:\gamma_i^t\left(\xbar,\sbar\right)<\eps_0}
\exp\left(2\lambda\gamma_i^t\left(\xbar,\sbar\right)\right)\\
&\ge
\eps_0\sum_{i=1}^m\exp\left(2\lambda\gamma_i^t\left(\xbar,\sbar\right)\right)
-\eps_0mn.
\end{align*}
Thus
\begin{align*}
\frac{\sum_{i=1}^m\exp\left(2\lambda \gamma_i^t \left(\xbar,\sbar\right)\right)
\cdot \gamma_i^t\left(\xbar,\sbar\right)}
{\left(\sum_{j=1}^m \exp\left(2\lambda\gamma_j^t\left(\xbar,\sbar\right)\right)\right)^{1/2}}
&\ge
\frac{\eps_0\sum_{i=1}^m\exp(2\lambda\gamma_i^t(\xbar,\sbar))}{\left(\sum_{j=1}^m \exp(2\lambda\gamma_j^t(\xbar,\sbar))\right)^{1/2}}-\frac{\eps_0mn}
{\left(\sum_{j=1}^m \exp(2\lambda\gamma_j^t(\xbar,\sbar))\right)^{1/2}}\\
&\ge\eps_0\left(\sum_{i=1}^m
\exp\left(2\lambda\gamma_i^t\left(\xbar,\sbar\right)\right)\right)^{1/2}
-\eps_0\sqrt{m}n,
\end{align*}
so, recalling $2\lambda\eps_0=\log n$,
\begin{equation}
\Psi^t\left(x^\new,s^\new\right)\le\Psi^t\left(x,s\right)
-\eps_0\left(\sum_{i=1}^m\exp\left(2\lambda\gamma_i^t\left(\xbar,\sbar\right)\right)\right)^{1/2}
+
\alpha\sqrt{m}n\log n.
\end{equation}

We now deal with the effects of approximating $x$ and $s$ with $\xbar$ and $\sbar$.
By \Cref{lem:gamma_approx}, we have $\gamma_i^t(\xbar,\sbar)\ge\gamma_i^t(x,s)-10\beta\eps$.
Recall $\beta=10\alpha=10\eps/(C_K\lambda)$, so we have
\begin{align*}
\exp\left(2\lambda\gamma_i^t\left(\xbar,\sbar\right)\right)
&\ge
\exp\left(2\lambda\gamma_i^t\left(x,s\right)\right)
\exp\left(-20\lambda\beta\eps\right)\\
&\ge
\exp\left(2\lambda\gamma_i^t\left(x,s\right)\right)
\exp\left(-200\eps^2/C_K\right),
\end{align*}
so our entire sum is reduced only by a negligible constant factor for small enough $\eps$.
That is, for some $\eps_1$,
\begin{equation}\label{eq:potential_no_t}
\Psi^t\left(x^\new,s^\new\right)
\le
\Psi^t\left(x,s\right)
-\eps_1\left(\sum_{i=1}^m\exp\left(2\lambda\gamma_i^t\left(x,s\right)\right)\right)^{1/2}
+
\alpha\sqrt{m}n\log n.
\end{equation}

We now consider the effects of taking $t\gets(1-\eta)t$.
We have
\[
\mu_i^{\that}\left(x^\new,s^\new\right)
=
\mu_i^t\left(x^\new,s^\new\right)
+
\eta\frac{s_i}{t} + O\left(\eta^2\frac{s_i}{t}\right).
\]
The $O(\eta^2)$ term is on the scale of $\alpha^2\eps^2$ and can be ignored.
We rewrite $\eta\frac{s_i}{t}$ as follows
\[
\eta\frac{s_i}{t}
=
\eta\mu_i^t\left(x^\new,s^\new\right)-\eta\nabla\phi_i\left(x_i\right).
\]
Then let
\[
\delta_{\mu,i}
\coloneqq
\eta\mu_i^t\left(x^\new,s^\new\right)
-
\eta\nabla\phi_i\left(x_i\right).
\]
We again analyze the effect on $\gamma_i^t(x,s)$.
Following the proof of \Cref{lem:potential},
\begin{align*}
\gamma_i^{\that}\left(x^\new,s^\new\right)
&=
\gamma_i^t\left(s,x\right)
+2\delta_{\mu,i}^\top\nabla^2\phi_i\left(x_i\right)^{-1}\mu_i^t(x,s)
+\delta_{\mu,i}^\top\nabla^2\phi_i\left(x_i\right)^{-1}\delta_{\mu,i}
+O\left(\eps^2\right)\\
&\approx_{\eps^2}
\gamma_i^t\left(s,x\right)+2\delta_{\mu,i}^\top\nabla^2\phi_i\left(x_i\right)^{-1}\mu_i^t\left(x,s\right)\\
&=
\gamma_i^t\left(s,x\right)
+2\left(\eta\mu_i^t\left(x^\new,s^\new\right)-\eta\nabla\phi_i\left(x_i\right)\right)^\top
\nabla^2\phi_i\left(x_i\right)^{-1}\mu_i^t\left(x,s\right)\\
&\approx_{\eps^2}
\gamma_i^t\left(s,x\right)+2\eta\mu_i^t\left(x,s\right)^\top
\nabla^2\phi_i\left(x_i\right)^{-1}\mu_i^t\left(x,s\right)
-2\eta\nabla\phi_i\left(x_i\right)^\top\nabla^2\phi_i\left(x_i\right)^{-1}\mu_i^t\left(x,s\right)\\
&\le
\gamma_i^t\left(s,x\right)
+2\eta\gamma_i^t\left(x,s\right)
+2\eta\norm{\nabla\phi_i\left(x_i\right)}_{\nabla^2\phi_i\left(x_i\right)^{-1}}
\norm{\mu_i^t\left(x,s\right)}_{\nabla^2\phi_i\left(x_i\right)^{-1}}\\
&\le
\gamma_i^t\left(x,s\right) + \left(2\eta+2\eta\sqrt{\nu_i}\right)
\gamma_i^t\left(x,s\right)^{1/2}\\
&\le
\gamma_i^t\left(x,s\right)
+
3\eta\sqrt{\nu_i}\gamma_i^t\left(x,s\right)^{1/2}.
\end{align*}
We again use \Cref{lem:exp_analysis} on each block and take the sum.
Recall we chose $\eta=\frac{\eps\alpha}{C_{center}\sqrt{\nu}}$;
thus the quadratic term is again on the scale of $\eps^3$.
To the first order, our change is
\begin{align*}        \sum_{i=1}^m\exp\left(\lambda\gamma_i^t\left(x,s\right)\right)
3\lambda\eta\sqrt{\nu_i}\gamma_i^t\left(x,s\right)^{1/2}
&=
\frac{3\lambda\alpha\eps}{C_{center}}
\sum_{i=1}^m\exp\left(\lambda\gamma_i^t\left(x,s\right)\right)
\gamma_i^t\left(x,s\right)^{1/2}\frac{\sqrt{\nu_i}}{\sqrt{\nu}}\\
&\le
\frac{3\lambda\alpha\eps^2}{C_{center}}
\sum_{i=1}^m\exp\left(\lambda\gamma_i^t\left(x,s\right)\right)
\frac{\sqrt{\nu_i}}{\sqrt{\nu}}\\
&\le
\frac{3\lambda\alpha\eps^2}{C_{center}}
\left(\sum_{i=1}^m\exp\left(2\lambda\gamma_i^t\left(x,s\right)\right)\right)^{1/2}
\left(\sum_{i=1}^m\frac{\nu_i}{\nu}\right)^{1/2}\\
&\le
\frac{3\lambda\alpha\eps}{C_{center}}
\cdot
\left(\sum_{i=1}^m\exp\left(2\lambda\gamma_i^t\left(x,s\right)\right)\right)^{1/2}
\end{align*}
Combining this with \eqref{eq:potential_no_t} gives us that
\begin{align*}
\Psi^{\that}\left(x^\new,s^\new\right)
&\le
\Psi^t\left(x,s\right)
-\eps_1\left(\sum_{i=1}^m\exp\left(2\lambda\gamma_i^t\left(x,s\right)\right)\right)^{1/2}
\\ & \qquad
+\frac{3\lambda\alpha\eps}{C_{center}}
\cdot
\left(\sum_{i=1}^m\exp\left(2\lambda\gamma_i^t\left(x,s\right)\right)\right)^{1/2}
+\alpha\sqrt{m}n\log n \\
&\le
\left(1 - \frac{\eps\alpha}{C_{center}^2\sqrt{\nu}}\right)
\Psi^t\left(x, s\right) + O\left(n^2\right),
\end{align*}
where we used
$\sqrt{m} (\sum_{i=1}^m\exp(2\lambda\gamma_i^t(x,s)))^{1/2}
\ge
\Psi^t\left(x, s\right)$
by Cauchy-Schwarz, or approximation ratio between $\ell_{2}$ and $\ell_{1}$ norms.
\end{proof}

\begin{lemma}[Basic Step Size Bounds]
\label{lem:step_size}
In \Cref{alg:shortstep} we have that 
\begin{enumerate}
\item $\|g\|_2\le\alpha\eps$,\label{item:step_size_1}
\item $\|\delta_1\|_2\le\alpha\eps\exp(\eps)$, and\label{item:step_size_2}
\item $\|\delta_2\|_2\le\alpha\eps\exp(3\eps/2)$.\label{item:step_size_3}
\end{enumerate}
\end{lemma}

\begin{proof}
(\cref{item:step_size_1})
Recall $g=\alpha g^t(\xbar,\sbar)$, as defined in \Cref{def:gradient}. Then
\begin{align*}
\norm{g}_2^2
&=
\sum_{i=1}^m\norm{\alpha g_i}_2^2\\
&=
\alpha^2\sum_{i=1}^m
\frac{\exp\left(2\lambda \gamma_i^t\left(\xbar,\sbar\right)\right)
\cdot
\norm{\g^2\phi\left(\xbar_i\right)^{-1/2}\mu_i^t\left(\xbar,\sbar\right)}_2^2}
{\sum_{i=1}^m \exp\left(2\lambda\gamma_i^t\left(\xbar,\sbar\right)\right)}\\
&=
\alpha^2\frac{\sum_{i=1}^m \exp\left(2\lambda\gamma_i^t\left(\xbar,\sbar\right)\right)
\cdot\gamma_i^t\left(\xbar,\sbar\right)}
{\sum_{i=1}^m \exp\left(2\lambda\gamma_i^t\left(\xbar,\sbar\right)\right)}\\
&\lesssim
\alpha^2\eps^2
\end{align*}
where we used the fact that $x$ and $s$ are centered (\Cref{def:wellcenter}), and that $\xbar$ and $\sbar$ are $\beta$-approximations of $x$ and $s$ (\Cref{lem:gamma_approx}).
    
(\cref{item:step_size_2})
$\delta_1 = \g^2 \Phi(\xbar)^{-1/2}AH^{-1}A^\top\nabla^2\Phi(\xbar)^{-1/2} g$
and
$H \approx_{\alpha} A^\top \g^2 \Phi(\xbar)^{-1} A$
combine to give
\[
\norm{\delta_1}_2
=
\norm{H^{-1}A^\top\nabla^2\Phi(\xbar)^{-1/2} g}_{A^\top \g^2 \Phi(\xbar)^{-1} A}.
\]
Applying $H \approx_{\alpha} A^\top \g^2 \Phi(\xbar)^{-1} A$
and $\alpha\le\eps$ twice,
\begin{align*} 
\norm{H^{-1}A^\top\nabla^2\Phi(\xbar)^{-1/2} g}_{A^\top \g^2 \Phi\left(\xbar\right)^{-1} A}
&\le
\exp\left(\eps/2\right)
\norm{A^\top\nabla^2\Phi\left(\xbar\right)^{-1/2} g}_{H^{-1}}\\
&\le
\exp\left(\eps\right)
\norm{g}_{\g^2 \Phi\left(\xbar\right)^{-1/2}
A\left(A^\top \g^2 \Phi(\xbar)^{-1} A\right)^{-1}
A^\top\nabla^2\Phi\left(\xbar\right)^{-1/2}}\\
        &\le\exp(\eps)\|g\|_2 \le\alpha\eps\exp(\eps)
    \end{align*}
    where the last inequality comes from $\g^2 \Phi(\xbar)^{-1/2}A(A^\top \g^2 \Phi(\xbar)^{-1} A)^{-1}A^\top\nabla^2\Phi(\xbar)^{-1/2}$ being a projection matrix.

(\cref{item:step_size_3})
Next, recall $\delta_2 =\g^2\Phi(\xbar)^{-1/2}AH^{-1}(A^\top x - b)$ and $\xbar\approx_\beta x$. Then
\[
\norm{\delta_2}_2
=
\norm{H^{-1}\left(A^\top x - b\right)}_{A^\top \g^2 \Phi\left(\xbar\right)^{-1} A}.
\]
$\xbar\approx_\beta x$ and $\beta\le\eps$ tells us
$A^\top \g^2 \Phi(\xbar)^{-1} A\approx_\eps A^\top \g^2 \Phi(x)^{-1} A$,
and so we also have $H \approx_{2\eps} A^\top \g^2 \Phi(x)^{-1} A$.
Applying these,
\begin{align*}
\norm{H^{-1}\left(A^\top x - b\right)}
_{A^\top \g^2 \Phi\left(\xbar\right)^{-1} A}
&\le
\exp\left(\frac{\eps}{2}\right)\norm{A^\top x - b}_{H^{-1}}\\
&\le
\exp\left(\frac{3\eps}{2}\right)
\norm{A^\top x - b}_{(A^\top \g^2 \Phi\left(x\right)^{-1} A)^{-1}}\\
&\le
\alpha\eps\exp\left(\frac{3\eps}{2}\right),
\end{align*}
where the last inequality comes from $(x,s)$ being centered (\Cref{def:wellcenter}).
\end{proof}

\begin{lemma}\label{lem:stability}
    In \Cref{alg:shortstep} we have
    \begin{itemize}
        \item $\|\nabla^2\Phi(\xbar)^{1/2}\E[\delta_x]\|_2\le6\alpha\eps$,
        \item $\|\nabla^2\Phi(x)^{1/2}\bar{\delta}_x\|_{\infty,2}\lesssim\eps/(C^2\lambda)$ with high probability,
        \item $\|\E[\nabla^2\Phi(x)\bar{\delta}_x^2]\|_2\lesssim\alpha\eps$.
    \end{itemize}
\end{lemma}

\begin{proof} We prove the three items individually.

\paragraph{Bound on $\|\nabla^2\Phi(\xbar)^{1/2}\E[\delta_x]\|_2$}

By definition, the triangle inequality, then \Cref{lem:step_size} we have
\[
\norm{\nabla^2\Phi\left(\xbar\right)^{-1/2}
\E\left[\delta_x\right]}_2
=
\norm{g-\delta_1-\delta_2}_2
\le
\norm{g}_2 + \norm{\delta_1}_2 + \norm{\delta_2}_2
\le
6\alpha\eps
\]
for $\eps<1/80$.

\paragraph{Bound on $\|\nabla^2\Phi(x)^{1/2}\bar{\delta}_x\|_{\infty,2}$.}

Observe by self-concordance and $\eps$-approximation (\Cref{def:epsapprox})
\begin{multline*} \norm{\nabla^2\Phi\left(x\right)^{1/2}\bar{\delta}_x}_{\infty,2}
=
\norm{\nabla^2\Phi\left(x\right)^{1/2}
\nabla^2\Phi\left(\xbar\right)^{-1/2}
\left(g - R\delta_r\right)}_{\infty,2}\\
\lesssim
\norm{g - R\delta_r}_{\infty,2}
=
\norm{g - \delta_r - \left(R\delta_r - \delta_r\right)}_{\infty,2}\\
\le
\norm{g-\delta_r}_2  + \norm{R\delta_r - \delta_r}_{\infty,2}.
\end{multline*}
The first term is exactly the term bounded above, which is $6\alpha\eps$.
We further have
\[
\norm{R\delta_r-\delta_r}_{\infty,2}
\le\sqrt{C_K}\alpha/C^2
\le\eps/(C^2\lambda)
\]
by the (Maximum) property of \cref{def:valid}, \cref{item:valid_maximum},
where recall that $\alpha=\frac{\eps}{C_K\lambda}$.

\paragraph{Bound on $\|\E[\nabla^2\Phi(x)\bar{\delta}_x^2]\|_2$.}

Again we use self-concordance and $\eps$-approximation:
\begin{align*}
\norm{\E\left[\nabla^2\Phi\left(x\right)\bar{\delta}_x^2\right]}_2 
&\lesssim
\norm{\nabla^2\Phi\left(\xbar\right)\E\left[\bar{\delta}_x^2\right]}_2\\
&\lesssim
\norm{\E\left[g - R\delta_r\right]^2}_2\\
&\lesssim
\norm{g^2}_2
+
\norm{E\left[R^2\delta_r^2\right]}_2.
\end{align*}
The first term can be bounded above by $\alpha\eps$ using \Cref{lem:step_size},
while the second term can be bounded using the Variance condition and \Cref{lem:step_size} again:
\[
\norm{\E\left[R^2\delta_r^2\right]}_2
\le
\norm{\delta_r^2}_2 + \frac{\alpha}{C^2} \norm{\delta_r}_2<\alpha\eps.
\]
\end{proof}

\begin{lemma}[Feasibility analysis]
\label{lem:feasible}
Assume that $(x, s)$ are $\eps$ well-centered at path parameter $t$. Then with high probability,
\[ \|A^\top x^{\new} - b\|_{\left(A^\top \g^2\Phi(x^{\new})^{-1} A\right)^{-1}} \le \alpha^2. \]
\end{lemma}

The proof relies on the following fact about matrix approximations.

\begin{lemma}
[\cite{BLNPSSSW20}, Lemma 4.32 in~\url{https://arxiv.org/pdf/2009.01802v2}]
\label{lem:difference_norm}
If $M\approx_\eps N$ for symmetric PD $M,N\in\R^{n\times n}$ and $\eps\in[0,1/2)$ then
\[\|N^{-1/2}(M-N)N^{-1/2}\|_2\le\eps+\eps^2.\]
\end{lemma}

\begin{proof}(of \Cref{lem:feasible})
Recall $x^\new=x+\nabla^2\Phi(\xbar)^{-1/2}(g-R\delta_r)$ where $\delta_r=\delta_1+\delta_2$. Further, recall that
\begin{align*}
\delta_1 &= \g^2 \Phi\left(\xbar\right)^{-1/2}AH^{-1}A^\top
\nabla^2\Phi\left(\xbar\right)^{-1/2} g \\
\delta_2 &= \g^2\Phi\left(\xbar\right)^{-1/2}AH^{-1}\left(A^\top x - b\right).
\end{align*}
Then we can define the local variable
\[
d\coloneqq A^\top\nabla^2\Phi\left(\xbar\right)^{-1/2}g+A^\top x-b.
\]
and rewrite our step as
\[
\delta_r
=
\nabla^2\Phi\left(\xbar\right)^{-1/2}AH^{-1}d
\]

Now, consider the idealized step $x^*$ where there is no matrix approximation error, i.e.
\[
H= A^\top \g^2 \Phi\left(\xbar\right)^{-1} A,
\]
and no sampling error, i.e. $I=R$. Formally, $x^*\coloneqq x+\delta_x^*$ where $\delta_x^*\coloneqq \nabla^2\Phi(\xbar)^{-1/2}(g-\delta_r^*)$ and $\delta_r^*\coloneqq \nabla^2\Phi(\xbar)^{-1/2}A(A^\top \g^2 \Phi(\xbar)^{-1} A)^{-1}d.$ Now,
\begin{align*}
A^\top x^*
&=
A^\top x+A^\top \nabla^2\Phi\left(\xbar\right)^{-1/2}\left(g-\nabla^2\Phi\left(\xbar\right)^{-1/2}A\left(A^\top \g^2 \Phi\left(\xbar\right)^{-1} A\right)^{-1}d\right)\\
&=
A^\top x+A^\top \nabla^2\Phi\left(\xbar\right)^{-1/2}g-d\\
&=
A^\top x+A^\top \nabla^2\Phi\left(\xbar\right)^{-1/2}g
-A^\top\nabla^2\Phi\left(\xbar\right)^{-1/2}g - \left(A^\top x-b\right)\\
&=b.
\end{align*}
Thus, we see that $x^*$ obeys the linear constraints for feasibility. 
Consequently, it suffices to bound the error induced by matrix approximation error and sampling error, i.e.
\begin{align}{\label{eq:feas_error}}
A^\top x^\new - b
&=
A^\top\left(x^\new-x^*\right)
\notag \\&=
A^\top\nabla^2\Phi\left(\xbar\right)^{-1/2}
\left(\delta_r^* - R\delta_r\right)\notag\\
&=
A^\top\nabla^2\Phi\left(\xbar\right)^{-1/2}
\left(\nabla^2\Phi\left(\xbar\right)^{-1/2}A
\left(A^\top \g^2 \Phi\left(\xbar\right)^{-1} A\right)^{-1}
-R\nabla^2\Phi\left(\xbar\right)^{-1/2}AH^{-1}\right)d
\notag\\&=
\left(I-A^\top\nabla^2\Phi(\xbar)^{-1/2}R
\nabla^2\Phi\left(\xbar\right)^{-1/2}AH^{-1}\right)d
\end{align}
Now, by \Cref{def:valid} item 6 (spectral approximation)
and definition of our algorithm, we have with high probability
\[
A^\top\nabla^2\Phi\left(\xbar\right)^{-1/2}R
\nabla^2\Phi\left(\xbar\right)^{-1/2}A
\approx_\alpha
A^\top\nabla^2\Phi\left(\xbar\right)^{-1}A\approx_\alpha H,
\]
so
\begin{align*}
&\norm{\left(A^\top\nabla^2\Phi\left(\xbar\right)^{-1}A\right)^{-1/2}
\left(I - A^\top\nabla^2\Phi\left(\xbar\right)^{-1/2}R
\nabla^2\Phi\left(\xbar\right)^{-1/2}AH^{-1}\right)H^{1/2}}_2\\
&\quad =
\norm{\left(A^\top\nabla^2\Phi\left(\xbar\right)^{-1}A\right)^{-1/2}
\left(H^{-1/2}-A^\top\nabla^2\Phi\left(\xbar\right)^{-1/2}R
\nabla^2\Phi\left(\xbar\right)^{-1/2}A\right)H^{-1/2}}_2\\
&\quad \le
\exp\left(\alpha\right)
\norm{H^{-1/2}\left(H^{-1/2}
-
A^\top\nabla^2\Phi\left(\xbar\right)^{-1/2}R
\nabla^2\Phi\left(\xbar\right)^{-1/2}A\right)H^{-1/2}}_2\\
&\quad \le
3\alpha.
\end{align*}
where we used \cref{lem:difference_norm}
and the fact that $(2\alpha+4\alpha^2)\exp(\alpha)\le3\alpha$. Consequently, combining with \eqref{eq:feas_error} yields that
\begin{multline*}
\norm{A^\top x^\new - b}
_{\left(A^\top\nabla^2\Phi\left(\xbar\right)^{-1}A\right)^{-1}} \\
\quad =
\norm{\left(A^\top\nabla^2\Phi\left(\xbar\right)^{-1}A\right)^{-1/2}
\left(I A^\top\nabla^2\Phi\left(\xbar\right)^{-1/2}R\nabla^
\Phi\left(\xbar\right)^{-1/2}AH^{-1}\right)H^{1/2}H^{-1/2}d}_2\\
\quad \le
3\alpha\norm{d}_{H^{-1}}
\le
3\alpha \left(\norm{A^\top\nabla^2\Phi(\xbar)^{-1/2}g}_{H^{-1}}
+\norm{A^\top x-b}_{H^{-1}}\right).
\end{multline*}
For the first term, by \Cref{lem:step_size} we have:
\begin{align*}
\norm{A^\top\nabla^2\Phi\left(\xbar\right)^{-1/2}g}_{H^{-1}}
&\le
\exp\left(\frac{\alpha}{2}\right)\norm{g}
_{\g^2 \Phi\left(\xbar\right)^{-1/2}A\left(A^\top \g^2 \Phi\left(\xbar\right)^{-1} A\right)^{-1}
A^\top\nabla^2\Phi\left(\xbar\right)^{-1/2}}\\
&\le\exp\left(\frac{\alpha}{2}\right)\norm{g}_2
\le
4\eps\alpha.
\end{align*}
We bound the second term using the approximate feasibility of our original point:
\[
\norm{A^\top x  - b}_{H^{-1}}
\le
\exp\left(\gamma\right)
\norm{A^\top x - b}_{\left(A^\top \g^2\Phi\left(x\right)^{-1} A\right)^{-1}}
\le
\exp\left(\frac{\gamma}{2}\right)\alpha\eps.
\]
Combining yields that
    \[ \|A^\top x^\new b\|_{(A^\top\nabla^2\Phi(\xbar)^{-1}A)^{-1}} \le 3\alpha(4\alpha\eps+2\beta\eps)\le 20\alpha^2\eps\le0.25\alpha^2. \]
    Finally, by self-concordance and $\eps$-approximation (\Cref{def:epsapprox}),
    \[ \|A^\top x^\new b\|_{(A^\top\nabla^2\Phi(x)^{-1}A)^{-1}}\le\alpha^2.\]
\end{proof}
\subsection{Stability Bounds}

In this section we prove that $x$ changes slowly, i.e., bound the number of times that a coordinate of $x$ may change enough that $\xbar$ must be updated. $s$ also changes slowly, but that follows because $\|\g^2 \Phi(x)^{1/2} \bar{\delta}_s\|_2 \le 1$ by the definition of $\bar{\delta}_s$ and $g$.
For $x$, the proof is more complicated because the random matrix $R$ which is being used to subsample the change at each iteration. However, using a martingale/potential argument based on previous works (see eg. \cite[Lemma 4.44]{BLLSSSW21}) we can argue that there is a stable subsequence of the sequence of $x$ vectors which changes slowly.

\begin{lemma}
\label{lem:xstable}
Let $(x^{(k)}, s^{(k)})$ for $k \in [T]$ be a sequence of points found by repeatedly calling short-step (see \cref{alg:shortstep}). With high probability, there is a sequence of points $\hat{x}^{(k)}$ such that:
\begin{enumerate}
    \item (Nearby) For all $k \in [T]$ and $i \in [m]$ it holds that
    \[ \|\g^2\phi_i(x^{(k)}_i)^{1/2}(x^{(k)}_i - \hat{x}^{(k)}_i)\|_2 \le \alpha/2, \enspace \text{ and} \]
    \item ($\ell_2$-Stability) For all $k \in [T]$ it holds that $\|\hat{x}^{(k+1)} - \hat{x}^{(k)}\|_{\g^2\Phi(x^{(k)})} \le 2\alpha$.
\end{enumerate}
\end{lemma}
\begin{proof}
    Define $\hat{x}^{(1)}=x^{(1)}$. Define the \textit{stability potential,} analogous to the centrality potential in \eqref{eq:psi}, as
    \[ \Psi_\stab(x,\hat{x})\coloneqq \sum_{i=1}^m\exp\left(\lambda_\stab\gamma_i\right) \]
    for $\gamma_i=\|\hat{x}_i-x_i\|_{\nabla^2\phi_i(x_i)}^2$ and $\lambda_\stab=C\log(n^2)/\alpha$ for sufficiently large $C$.
    
    We take steps to decrease this stability potential following gradient descent. Specifically, our gradient is
\[
g_{\stab,i}^{\left(k\right)}
\coloneqq\frac{\exp\left(\lambda_\stab\gamma_i^{\left(k\right)}\right)
\cdot
\g^2\phi_i\left(x_i^{\left(k\right)}\right)^{1/2}
\left(\hat{x}_i^{\left(k\right)}-x_i^{\left(k\right)}\right)}
{\left(\sum_{i=1}^m\exp\left(2\lambda_\stab\gamma_i^{\left(k\right)}\right)\right)^{1/2}}.
\]
Define $\delta_{\hat{x}}^{(k)}\coloneqq\alpha g_\stab$ and
\[
\hat{x}^{\left(k+1\right)}
\coloneqq
\hat{x}^{\left(k\right)}
-
\E\left[\bar{\delta}_x\right]
-
\g^2\Phi\left(x^{\left(k\right)}\right)^{-1/2}
\delta_{\hat{x}}.
\]
We first show $\ell_2$-stability.
\begin{align*}
\norm{\hat{x}^{\left(k+1\right)} - \hat{x}^{(k)}}
_{\g^2\Phi\left(x^{\left(k\right)}\right)}
&=
\norm{\E\left[\bar{\delta}_x\right]
+\nabla^2\Phi\left(x^{\left(k\right)}\right)^{-1/2} \delta_{\hat{x}}}_{\g^2\Phi\left(x^{\left(k\right)}\right)}\\
&\le
\norm{\E\left[\bar{\delta}_x\right]}_{\g^2\Phi\left(x^{\left(k\right)}\right)}
+
\norm{\nabla^2\Phi\left(x^{\left(k\right)}\right)^{-1/2}\delta_{\hat{x}}}
_{\g^2\Phi\left(x^{\left(k\right)}\right)}\\
&\le7
\alpha\eps+\norm{\delta_{\hat{x}}}_2\\
&=
7\alpha\eps+\alpha\left(\sum_{i=1}^m
\frac{\exp\left(\lambda_\stab\gamma_i^{\left(k\right)}\right)
\norm{\g^2\phi_i\left(x_i^{\left(k\right)}\right)^{1/2}
\left(\hat{x}_i^{\left(k\right)}-x_i^{\left(k\right)}\right)}_2^2}
{\sum_{i=1}^m\exp\left(2\lambda_\stab\gamma_i^{\left(k\right)}\right)}\right)^{1/2}\\
    &\le8\alpha\eps
\end{align*}
where we use by induction
\[
\norm{\g^2\phi_i\left(x^{\left(k\right)}_i\right)^{1/2}
\left(x^{\left(k\right)}_i - \hat{x}^{\left(k\right)}_i\right)}_2
\le \frac{\alpha}{2}
\]
and $\alpha<\eps$.
Note that this is not circular as we use the nearby property of $x^{(k)}_i$ and $\hat{x}^{(k)}_i$ to prove $\ell_2$ stability of $\hat{x}^{(k+1)}$ and $\hat{x}^{(k)}$,
and we now show how to use that to find the nearby property for $x^{(k+1)}_i$ and $\hat{x}^{(k+1)}_i$.

To do so, we use \Cref{lem:potential} and show that our stability potential is never large.
To be able to use the lemma, we observe
\begin{align*}
\norm{\E\left[\left(\hat{x}_i^{\left(k\right)}-x_i^{\left(k\right)}\right)-\left(\hat{x}_i^{\left(k+1\right)}-x_i^{\left(k+1\right)}\right)\right]
}
_{\nabla^2\phi_i\left(x_i^{\left(k\right)}\right)}
&=
\norm{\hat{x}_i^{\left(k\right)}
-\hat{x}_i^{\left(k+1\right)}
-\E\left[\bar{\delta}_x\right]}_{\nabla^2\phi_i\left(x_i^{\left(k\right)}\right)}\\
&=
\norm{\nabla^2\Phi\left(x^{\left(k\right)}\right)^{-1/2}\delta_{\hat{x}}}_{\nabla^2\phi_i\left(x_i^{\left(k\right)}\right)}\\
&\le
\eps^2.
\end{align*}
Then, applying \Cref{lem:potential} to each block and summing gives
\begin{align*} \Psi_\stab\left(x^{\left(k+1\right)},\hat{x}^{\left(k+1\right)}\right)
&\le
\Psi_\stab\left(x^{\left(k\right)},\hat{x}^{\left(k\right)}\right)\\
&
\qquad -2\sum_{i=1}^m\exp\left(\lambda_\stab\gamma_i^{\left(k\right)}\right)
\lambda\left(\nabla^2\phi_i\left(x_i^{\left(k\right)}\right)^{-1/2}
\delta_{\hat{x},i}\right)^\top \nabla^2\phi_i\left(x_i^{\left(k\right)}\right)
\left(\hat{x}_i^{\left(k\right)}-x_i^{\left(k\right)}\right)\\
&
\qquad +4\sum_{i=1}^m
\exp\left(\lambda\gamma_i^{\left(k\right)}\right)
\exp\left(2\lambda\eps\left(\gamma_i^{\left(k\right)}\right)^{1/2}\right)
\lambda^2\eps^2\gamma_i^{\left(k\right)}.
\end{align*}
We have by induction $\gamma_i^{(k)}\le\alpha^2$;
thus again we can ignore the quadratic term.
We consider the first order term, without the sum or scaling for now:
\begin{align*}
&\E\left[\left(\nabla^2\phi_i\left(x_i^{\left(k\right)}\right)^{-1/2} \delta_{\hat{x},i}\right)^\top
\g^2\phi_i\left(x_i^{\left(k\right)}\right)
\left(\hat{x}_i^{\left(k\right)}-x_i^{\left(k\right)}\right)\right]\\
&\qquad=
\E\left[\delta_{\hat{x},i}\right]^\top
\g^2\phi_i\left(x_i^{\left(k\right)}\right)^{1/2}
\left(\hat{x}_i^{\left(k\right)}-x_i^{\left(k\right)}\right)\\
&\qquad=
\frac{\alpha\exp\left(\lambda_\stab\gamma_i^{\left(k\right)}\right)}{\left(\sum_{i=1}^m\exp\left(2\lambda_\stab\gamma_i^{\left(k\right)}\right)\right)^{1/2}}
\cdot\left(\hat{x}_i^{\left(k\right)}-x_i^{\left(k\right)}\right)^\top
\g^2\phi_i\left(x_i^{\left(k\right)}\right)
\left(\hat{x}_i^{\left(k\right)}-x_i^{\left(k\right)}\right)\\
&\qquad=
\frac{\alpha\exp\left(\lambda_\stab\gamma_i^{\left(k\right)}\right)\gamma_i^{\left(k\right)}}{\left(\sum_{i=1}^m\exp\left(2\lambda_\stab\gamma_i^{\left(k\right)}\right)\right)^{1/2}}.
\end{align*}
Thus our potential becomes
\[
\Psi_\stab\left(x^{\left(k+1\right)},\hat{x}^{\left(k+1\right)}\right)
\le
\Psi_\stab\left(x^{\left(k\right)},\hat{x}^{\left(k\right)}\right)
-
2\alpha\lambda_\stab
\sum_{i=1}^m
\frac{\exp\left(2\lambda_\stab\gamma_i^{\left(k\right)}\right) \gamma_i^{\left(k\right)}}
{\left(\sum_{i=1}^m\exp\left(2\lambda_\stab\gamma_i^{\left(k\right)}\right)\right)^{1/2}}.
\]

Let $\eps_0$ be such that $2\lambda_\stab\eps_0=\log n$. Following the proof of \Cref{lem:pot_drop}, 
\begin{align*}
\E\left[\Psi_\stab\left(x^{\left(k+1\right)},\hat{x}^{\left(k+1\right)}\right)\right]
&\le
\left(1-\frac{\alpha\log n}{\sqrt{m}}\right)
\Psi_\stab\left(x^{\left(k\right)},\hat{x}^{\left(k\right)}\right)
+\alpha\sqrt{m}n\log n.
\end{align*}
Recall $\Psi_\stab(x^{(1)},\hat{x}^{(1)})=m$. Then, by induction we have $\E[\Psi_\stab(x^{(k)},\hat{x}^{(k)})]\le nm$ for all $k$. Therefore with probability $1-n^{-12}$ we have that $\Psi_\stab(x^{(k)},\hat{x}^{(k)})\le n^{14}$ for all $k$. By the choice of $\lambda_\stab$ this implies that $\|\nabla^2\Phi(x^{(k)})^{1/2}(x^{(k)}-\hat{x}^{(k)})\|_{\infty,2}\le\alpha/2$, as desired.
\end{proof}

\section{Adaptive Sparsifier}
\label{sec:sparsifier}

In this section we establish \cref{thm:sparsify}. We start by establishing the following theorem on decremental sparsifiers, which can be extended to a fully-dynamic version with a standard binary bucketing scheme.

\begin{theorem}
\label{thm:decr}
Let $\kappa$ be a parameter that is $\poly(n, U)$,
and $A \in \R^{n \times d}$ be a matrix with row norms at most $\kappa$
($\|a_i\|_2^2 \leq \kappa$ for all $i$) undergoing row halving,
which is setting $a_i$ to $a_i/2$.
We can maintain in total time $\O(nd + d^{3+\omega})$ against an adaptive adversary leverage overestimates $\tautil_i$ of the rows of
$[A; \kappa^{-1/2} I]$ summing to $\Otil(d)$, that is:
\begin{itemize}
\item $\tautil_i \geq a_i^{\tp} (A^{\tp} A + \kappa^{-1} I)^{-1} a_i$ after every update.
\item $\sum_{i} \tautil_i \leq \Otil(d)$, where recall $\Otil(\cdot)$
hides poly-logarithmic factors of $n$, $U$, and thus $\kappa$.
\end{itemize}
\end{theorem}

We first deduce \cref{thm:sparsify} from \cref{thm:decr} via standard reductions, and then prove~\cref{thm:decr}.

\begin{proof}[Proof of \cref{thm:sparsify}]
We first implement row deletions by $O(\log \kappa)$ row halvings. Compared to \cref{thm:sparsify}, the leverage scores here are defined with respect to a slightly larger Gram matrix, including the extra $\kappa^{-1} I$ and the rows remained after halving. The remaining rows contribute at most $\frac{1}{\kappa^{10}} A^\tp A \preceq \kappa^{-9} I$. By $A^\top A \succeq \kappa^{-1} I$, we lose at most a constant factor of the leverage score overestimates for using row halving.

Next we reduce the fully dynamic version to this decremental version using a standard binary bucketing scheme. We maintain $O(\log Q)$ levels of decremental data structures for $Q$ batches of insertions, where $\mathcal{D}^{(\decr, k)}$ denotes the $k$-th level. At the $q$-th batch, we merge all the rows in $\mathcal{D}^{(\decr, 0)}, \mathcal{D}^{(\decr, 1)}, \dots, \mathcal{D}^{(\decr, \ell)}$ that are not deleted yet, together with the current batch, into the $\ell$-th level, where $\ell$ denotes the largest integer such that $2^{\ell}$ divides $q$. This clears all the levels below $\ell$, and reconstructs a decremental data structure at level $\ell$. 
All the deletions go to the corresponding level of the decremental data structure. Since the Gram matrix of each level sums to the total Gram matrix, the leverage score overestimates with respect to the Gram matrix of each level can be only larger.

The reduction blows up the parameters only by $O(\log Q)$ factor, which is hidden in $\Otil(\cdot)$. The total running time is $\Otil(Q \cdot d^{\omega+3} + n \cdot d)$, since there are a total of $\Otil(Q)$ decremental data structures with a total of $\Otil(n)$ rows.
\end{proof}

\subsection{Bounding Leverage Score Decreases}
\label{subsec:detbounds}

We show that the halve steps can be grouped into batches of constant leverage score total.
Each such batch ensures that the previous Gram matrix approximates the current Gram matrix.
Furthermore, the total number of such batches is readily boundable by a volume argument.

\begin{lemma}
\label{lem:perturbapprox}
Let $A$ be a matrix 
and $H$ a subset of rows whose total leverage score w.r.t $A^\tp A + \kappa^{-1} I$ 
is at most $1.2$, or equivalently, the removed leverage score is at most $0.9$:
\[
\sum_{i \in H} a_i^\tp \left( A^\tp A + \frac{I}{\kappa} \right)^{-1} a_i
-
\sum_{i \in H} \left(\frac{a_i}{2} \right)^\tp \left( A^\tp A + \frac{I}{\kappa} \right)^{-1} \left( \frac{a_i}{2} \right)
\leq 0.9,
\]
then the Gram matrix after halving the rows in $H$,
aka. replacing each $a_{i}$ with $\frac{1}{2} a_i$ for all $i\in H$,
$10$-approximates the one before
\[
A^\tp A + \frac{I}{\kappa}
\approx_{10}
\left(
    A^\tp A -
    \frac{3}{4} \left(A_{H, :}\right)^{\tp} A_{H, :}
\right)
+
\frac{I}{\kappa}.
\]
\end{lemma}

\begin{proof}

For any nonzero vector $x$, the ratio of the removed quadratic from to total quadratic form
\begin{align*}
    \frac{
    \frac{3}{4}
    x^\tp \left(A_{H, :}\right)^{\tp} A_{H, :}
    x
    }{
    x^\tp
    \left( A^\tp A + I/\kappa \right)
    x
    }
    =
    \sum_{i \in H}
    \frac{
    \frac{3}{4}
    x^\tp a_i^\tp a_i
    x
    }{
    x^\tp
    \left( A^\tp A + I/\kappa \right)
    x
    }
\le
    \frac{3}{4}
    \sum_{i \in H}
    a_i^{\tp}
    \left( A^\tp A + I/\kappa \right)^{-1}
    a_i
    \le 0.9
\end{align*}
by Rayleigh quotient inequality.
Therefore, $\frac{3}{4} \left(A_{H, :}\right)^{\tp} A_{H, :}
\preceq 0.9(A^\tp A + I/\kappa)$ and the lemma follows.
\end{proof}

The length of such a sequence of large leverage score deletion batches is also bounded.
For this we first define a matrix sequence created by a sequence of halvings.

\begin{definition}
\label{def:batchedhalvingsequence}
A batched halving sequence of an initial matrix $A = A^{(0)}$
of length $Q$ is a sequence of subsets of rows of $A$
$H^{(1)}, H^{(2)},  \ldots ,H^{(Q)} \subseteq [n]$ 
leading to the matrix sequence inductively for $1 \leq q \leq Q$ as
\[
A^{\left( q \right)}_i
= 
\begin{cases}
A^{\left( q - 1 \right)}_{i} & \qquad i \notin H^{\left( q\right)},\\
\frac{1}{2} A^{\left( q - 1\right)}_{i} & \qquad i \in H^{\left( q\right)},\\
\end{cases}
\]
with the property that for all $2 \le q \le Q$,
\[
\sum_{i \in H^{\left( q \right)} }
\left(a_{i}^{\left( q \right)}\right)^{\top}
\left( \left( A^{\left( q \right)}\right) ^{\tp} A^{\left( q \right)} + \frac{I}{\kappa} \right)^{-1}
a_{i}^{\left( q \right)}
\in
[0.001, 1.2].
\]
\end{definition}

\begin{lemma}
\label{lem:batchbound}
Let $A$ be a matrix such that that $\|a_i\|_2^2 \leq \kappa$,
and $H^{(1)}, H^{(2)},  \ldots ,H^{(Q)} \subseteq [n]$ be a batched
halving sequence (as defined in \cref{def:batchedhalvingsequence}),
it holds that $Q \leq \O(d)$.
\end{lemma}

\begin{proof}
For any positive-definite matrix $G$ and any vector $v$ with $v^\tp G^{-1} v = \tau < 1$,
by multiplicativeness of the determinant and the fact that $\det(I - XY) = \det(I - YX)$,
we get
\begin{align} \label{eq:det}
    \det(G - v v^\tp)
    & =
    \det( G^{1/2})
    \det(I - G^{-1/2} v v^\tp G^{-1/2})
    \det(G^{1/2}) \notag
    \\ & =
    \det(G) \det(I - v^\tp G^{-1} v)
    =
    \det(G) (1 - \tau) \le \det(G) e^{-\tau}.
\end{align}
Let $G^{(q)} := (A^{(q)})^\tp A^{(q)} + \frac{I}{\kappa}$. We transform $A^{(q)}$ to $A^{(q+1)}$ by halving one row at a time, and repeatedly apply the inequality \eqref{eq:det} to get 
\begin{align*}
    \det(G^{(q+1)}) < e^{-0.0001} \det(G^{(q)})
\end{align*}
since all intermediate Gram matrices (and thus the leverage scores with respect to the Gram matrices) are $10$-approximations. We conclude that $Q \le \log\left( \frac{\det(G^{(1)})}{\det(G^{(Q)})} \right) = O(d \log(nd\kappa)) = \Otil(d)$.

\end{proof}

Note that an immediate corollary of the above two facts is that the
total leverage score increases of all (remaining) rows during the course
of a deletion sequence is also $\O(d)$.
If we threshold leverage scores by additive $d / n$,
it suffices to sample all the rows to create approximate Gram matrices,
and the total increase still comes out to $\Otil(n)$ rows.
This is the primary motivation for our batching schemes.

\subsection{Checker-Induced Sequence}

By sketching the Gram matrix, we can create a checker
that in $\O(d)$ time estimates the leverage score of
row $a_i$ within a factor of $2$.

To find increases to leverage scores,
we utilize heavy hitter data structures.
There are two issues:
\begin{enumerate}
\item The Gram matrix is approximate, so any approximation error goes into the number of false positives.
\item The heavy hitter only works against an oblivious adversary,
so we need to hide decisions from the heavy hitter via a checker-induced
sequence, which is what we define below.
\end{enumerate}

\begin{definition}
\label{def:sequence}
A checker-induced leverage score estimation sequence for a halving sequence
deletion sequence $H^{(1)} \ldots H^{(Q)}$, where $Q < n$ and every row is halved for at most $O(\log n)$ times, is defined as
$\tautil^{(0)}$ setting to overestimates of initial leverage scores of $A^{(0)} = A$,
and repeatedly computed at each step $q \in [Q]$ as:
\begin{enumerate}
\item Create $\epscheck$-approximate Gram matrix $\Gtil^{(q)}$ by setting
\[
\epscheck \leftarrow \frac{0.1}{Q}
\]
and sampling the rows of $A^{(q)}$ with probabilities
$\tautil^{(q - 1)} \cdot O(\epscheck^{-2} \log{n})$
by \cref{lem:sample}.
\item Create fresh $O(\log{n}) \times d$ JL projection matrix $S$,
and use it to sketch the inverse Gram matrix
\[
Z^{\left( q \right)} \leftarrow S \left( \Gtil^{\left( q \right)} \right)^{-1/2}
\]
\item For each row $i$ halved, recompute its leverage score estimate using $Z^{(q)}$,
\[
\tautil_{i}^{\left( q \right)}
\leftarrow
\frac{d}{n}
+ 10 \norm{Z^{\left( q \right)} a^{\left( q \right)}_{i}}_2^2.
\]
\item For each integer $j$ such that $2^{j} | q$,
let $\qhat = q - 2^{j}$ be the other end of the
 dyadic-tiling aligned interval on batch numbers, and:
\begin{enumerate}
\item Create the matrix
\[
\Delta^{\left(\qhat, q\right)}
\coloneqq
\left( 1 + \epscheck \right) \left(\Gtil^{\left( q\right)}\right)^{-1}
-
\left( 1 - \epscheck \right) \left(\Gtil^{\left( \qhat \right)}\right)^{-1}
\]
along with sketch matrices
\[
Z^{\left(\widehat{q}, q \right)} \leftarrow
S \left( \Delta^{\left(\widehat{q}, q \right)} \right)^{1/2}
\]
\label{step:sketchdiffernece}
\item
\label{label:checkerpositive}
For each remaining row $a^{(q)}_i$ with large dot against $Z^{(\widehat{q}, q)}$, aka.
\[
\norm{Z^{(\widehat{q}, q)} a^{\left( q \right)}_{i} }_2^2 \geq \frac{d}{10 n \log{n}}
\]
recompute the leverage score estimate of that row using the sketch of the Gram matrix
\[
\tautil_{i}^{\left( q \right)}
\leftarrow
\frac{d}{n}
+ 10 \norm{Z^{\left( q \right)} a^{\left( q \right)}_{i}}_2^2.
\]
\end{enumerate}
\end{enumerate}
\end{definition}

We first verify that this checker-induced
leverage score overestimates $\tautil$ are indeed overestimates,
and sum to $\O(d)$.

\begin{lemma}
\label{lem:checkercorrect}
With high probability, after each batch $1 \leq q \leq Q$,
\[
\tautil^{(q)}_i
\geq \left(a^{\left( q \right)}_{i}\right)^{\top}
\left( \left( A^{\left( q \right)}\right)^{\tp} A^{\left( q \right)} + \frac{I}{\kappa} \right)^{-1}
a^{\left( q \right)}_{i}
\qquad
\forall q, i
\]
\end{lemma}

\begin{proof}

The proof is by induction on time $q$.
The case of $q = 0$ follows from $\tauhat^{(0)}$ being directly
initialized with overestimates.

Suppose at time $q$, some row's leverage score has increased by additive $\ge d/n$. 
Let $q_{last}$ be the last time this row's estimate was updated. 
The row remains unchanged between $[q_{last}, q]$, i.e., $a_i^{(q)} = a_i^{(q_{last})}$. We use $a_i$ to denote $a_i^{(q)}$, for brevity in this proof only. Then we have
\[
a_{i}^\tp \left( G^{\left(q\right)} \right)^{-1} a_{i}
-
a_{i}^\tp \left( G^{\left(q_{last}\right)} \right)^{-1} a_{i}
\geq
\frac{d}{n}
\]
Decomposing $[q_{last}, q]$ by dyadic tiling gives that there is
some $[q_l, q_r]$ such that
\[
a_{i}^\tp \left( G^{\left(q_r\right)} \right)^{-1} a_i
- a_{i}^\tp \left( G^{\left(q_{l}\right)} \right)^{-1} a_i
\geq
\frac{d}{n \log{n}}
\]
which combined with $(1 + \epscheck)(\Gtil^{(q_r)})^{-1} \succeq (G^{(q_r)})^{-1}$
and $ (G^{(q_l)})^{-1} \succeq (1 - \epscheck)(\Gtil^{(q_r)})^{-1}$ gives
\[
\frac{d}{n \log{n}}
\le
a_i^{\top} \left( 1 + \epscheck \right) \left(\Gtil^{\left( q_r\right)}\right)^{-1} a_i
-
a_i^{\top} \left( 1 - \epscheck \right) \left(\Gtil^{\left( q_{l}\right)}\right)^{-1} a_i
=
a_i^{\top} \Delta^{\left( q_l, q_r \right)} a_i
\]
Which means the sketch must have failed on the interval $[q_l, q_r]$.
Taking union bound over all $O(Q \log{Q}) \leq n^{O(1)}$ tiling intervals
and the sketches/samples of the Gram matrices themselves gives the overall guarantee.
\end{proof}

Note that the checking of each candidate $i$ takes time $\O(d)$.
So the important step is ensuring that only a small number candiates
are checked explicitly in creating this sequence.

\subsection{Locator via Heavy Hitter}

We use heavy hitters as a locator in \cref{label:checkerpositive} to locate a candidate list of rows efficiently. The locator/checker framework isolates the randomness of the heavy hitters from the adversary.

\begin{lemma}
\label{lem:locator}
In the checker-induced sequence as given in \cref{def:sequence}, we can use heavy hitter data structures to generate a list of candidates
that contains a superset of the rows identified in \cref{label:checkerpositive}
at every step $q \in [Q]$. The total size of the candidate list is bounded by $\Otil(n)$.
\end{lemma}

\begin{proof}
    We maintain the rows of $A$ by the heavy hitter data structure in parallel to the checker-induced sequence. This uses a total of $\Otil(n)$ $\textsc{Modify}$ operations.
    At step $q \in [Q]$, the heavy hitter calls $\textsc{Query}(Z^{(\widehat{q}, q)}, \frac{d}{20 n \log n})$ to generate a candidate list.
    The rows in \cref{label:checkerpositive} are then identified by enumerating over the candidate list and checking the inequality
    \[
    \norm{Z^{(\widehat{q}, q)} a^{\left( q \right)}_{i} }_2^2 \geq \frac{d}{10 n \log{n}}.
    \]

    It suffices to show the inputs to  \Cref{lem:heavy-hitter} are not adaptive, so the correctness follows and the candidate set is valid. We prove this using a \emph{simulation} argument. The checker-induced sequence is solely determined by the initial input $A$, the adversary, and the randomness of JL projection matrices in previous batches. The sequence can be simulated by a checker without using any heavy hitters. Therefore, the inputs generated by the adversary are independent with the randomness of the heavy hitter, concluding that the interface with the heavy hitter is non-adaptive.
    
    It remains to bound the total size of the candidate lists for heavy hitter.
    Plugging in the value of $\delta$ into the upper bound on the set returned from \cref{lem:heavy-hitter} gives
    \begin{align} \label{eq:candidate-list-size}
        \sum_{i=1}^{n} \sum_{(\widehat{q}, q)} 
        \norm{Z^{(\widehat{q}, q)} a^{\left( q \right)}_{i} }_2^2 \cdot \frac{20n \log n}{d}
        \le
        \Otil(n / d) \cdot 
        \sum_{i=1}^{n}
        \left( a^{\left( q \right)}_{i} \right)^\top
        \left(
        \sum_{(\widehat{q}, q)} 
        \Delta^{(\widehat{q}, q)}
        \right) 
        a^{\left( q \right)}_{i} .
    \end{align}
    By dyadic tiling, the sum of the approximated Gram matrices can be upper bounded by
    \begin{align*}
        \sum_{(\widehat{q}, q)} 
        \Delta^{(\widehat{q}, q)}
        \preceq (1+\epscheck)
        \left( \Gtil^{(Q)} \right)^{-1}
        +
        2 \log n \cdot \epscheck \sum_{q=1}^{Q-1} 
        \left( \Gtil^{(q)} \right)^{-1},
    \end{align*}
    and then we bound the sum of quadratic forms using the bound on the sum of leverage scores as
    \begin{align*}
       \sum_{i=1}^{n}
        \left( a^{\left( q \right)}_{i} \right)^\top
        \left(
        \sum_{(\widehat{q}, q)} 
        \Delta^{(\widehat{q}, q)}
        \right) 
        a^{\left( q \right)}_{i} 
        \le (1 + \epscheck) d + 2 \log n \cdot \epscheck \cdot Q d \le \Otil(d).
    \end{align*}
    Plugging it into \eqref{eq:candidate-list-size}, we conclude that the total size of the candidate list is bounded by $\Otil(n)$.    
\end{proof}

\subsection{Buffering to Form Batches}

We can now prove the overall decremental bound by buffering the
halving until their decreased leverage scores exceed a constant threshold.
This buffering preserves operator approximation by \cref{lem:perturbapprox},
and the total number of batches is bounded by \cref{lem:batchbound}.
We remark that the batches in the decremental data structure are created lazily by our algorithm, which is different from the batches given by the inputs in \cref{thm:sparsify}.

\begin{proof}(of \cref{thm:decr})
Build a buffer set $H_{buf}$ of the halving operations from the last batch $q_{last}$. We maintain the sum of leverage scores with respect to the sketched Gram matrix $Z^{(q_{last})}$ in the last batch. 

If the sum does not exceed $0.01$, by \cref{lem:perturbapprox}, the current Gram matrix is still $10$-approximated by $Z^{(q_{last})}$ so the leverage score overestimates are good enough.

Once the sum exceeds $0.01$, we create a new batch with all the halving operations in the buffer $H_{buf}$, and then clear the buffer and reset the sum. The total leverage score with respect to the previous batch $(A^{(q_{last})})^\tp A^{(q_{last})} + I/\kappa$ is at most $0.01 \cdot 10 + 1 = 1.1$ and at least $0.01 / 10 = 0.001$, which makes the batch valid as defined in \cref{def:batchedhalvingsequence}.

We can handle halving one row for multiple times in one batch by extending the set to a multiset and tracking the removed leverage scores. Alternatively, in the reduction from \cref{thm:sparsify} to \cref{thm:decr}, we can afford to pay $\Otil(Q_{out})$ extra batches to assure that one row is halved for at most once in one batch, where $Q_{out}$ denotes the number of outer batches in \cref{thm:sparsify}.

\paragraph{Running time.}  
By \cref{lem:locator}, the total size of the candidate list produced by the heavy hitter is $\Otil(n)$. 
The heavy hitter runs in $\Otil(nd)$ total time. 
The checker checks each row of the candidate list in $O(d)$ time, so the total time is also $\Otil(nd)$.

For each batch, the running time is dominated by creating the approximated Gram matrix $\Gtil^{(q)}$. The number of sampled rows is $\Otil(\epscheck^{-2} \cdot d) = \Otil(d^3)$. Computing the Gram matrix requires multiplying a $d \times d^3$ matrix with a $d^3 \times d$ matrix, which runs in $\Otil(d^{3 + \omega})$ time. This concludes the total running time $O(nd + d^{3 + \omega})$ since the number of batches is $\Otil(d)$ by \cref{lem:batchbound}.

\end{proof}

\section{Implementation and Runtime Analysis}

In this section we describe how to implement $\O(\sqrt{n})$ steps of the short-step IPM of \cref{sec:ipm} using the adaptive sparsifier data structure built in \cref{sec:sparsifier} as well as standard heavy-hitter data structures from prior works \cite{BLSS20,BLNPSSSW20}.

\subsection{Primal, Slack, and Gradient Maintenance}
\label{subsec:hhmaintain}

In this section we discuss how to efficiently maintain the vectors $\xbar$, $\sbar$, and $g = \alpha g^t(\xbar, \sbar)$ over the course of $\O(\sqrt{n})$ iterations of \cref{alg:shortstep}. We start by discussing $\sbar$ which is mostly a simple adaptation of previous works \cite{BLSS20,BLNPSSSW20,BLLSSSW21} which uses $\ell_2$-heavy hitters. Then we discuss $\xbar$, which amounts to discussing how to efficiently sample the matrix $R$ to be \emph{valid} (see \cref{def:valid}). Finally, we discuss how to maintain $g$, which is simple given a list of explicit changes to the $\xbar$ and $\sbar$ vectors.

The approximation $\sbar$ can be maintained using the following general lemma, which is an adaptation of \cite[Theorem E.1]{BLLSSSW21}.
\begin{lemma}[Slack maintenance]
\label{lem:slack}
Let $A \in \R^{n \times d}$ be a matrix, $s \in \R^n$ initially be $s \assign \vec{0}$, and $D \in \R^{n \times n}$ be a positive definite block-diagonal matrix, where for $i \in [m]$ we denote the $i$-th block as $D_i \in \R^{n_i \times n_i}$ and $\sum_{i \in [m]} n_i = n$. Let $M = \max_{i \in [m]} n_i$. There is a data structure that supports the following operations, with the following runtimes.
\begin{itemize}
    \item $\textsc{UpdateScaling}(i, M \in \R^{n_i \times n_i})$. Set the $i$-th block of $D$ to $M$, i.e., $D_i \assign M$.
    \item $\textsc{UpdateSlack}(h \in \R^d)$. Set $s \assign s + Ah$, where it is guaranteed that $\|DAh\|_2 \le 1$.
\end{itemize}
The algorithm maintains a vector $\sbar \in \R^n$ satisfying $\|D_i(s_i - \sbar_i)\|_2 \le \eps$ for all $i \in [m]$, and reports changes to $\sbar$ explicitly after each operation. The algorithm succeeds with high probability against an adaptive adversary, with initialization times $\O(\eps^{-2}nd)$ and:
\begin{itemize}
    \item The amortized update time of $\textsc{UpdateScaling}$ is $\O(d)$, and updates $\sbar$ in one coordinate, and
    \item After the $j$-th call to $\textsc{UpdateSlack}$, the algorithm updates $\sbar$ in at most $\O(\eps^{-2}2^{2 v_2(j)})$ coordinates with total update time $\O(\eps^{-2} d \cdot 2^{2v_2(j)})$.
\end{itemize}
\end{lemma}
\begin{proof}
Let us start by defining the algorithm. Define $\beps = \frac{\eps}{2 \log n}$ and for $k \in \mathbb{Z}_{\ge0}$ such that $2^k \le \sqrt{n}$ define $\beps_k \coloneqq \frac{\beps}{5M \cdot 2^k}$. Use \cref{thm:hh} to define heavy hitter matrices $Q_k \in \R^{\O(\beps_k^{-2}) \times n}$. The algorithm will also maintain matrices $D^{(k)} \in \R^{n \times n}$ defined as follows. Let $t$ be the current total number of calls to $\textsc{UpdateSlack}$ and let $\that = 2^k \lfloor t/2^k \rfloor$, i.e., the largest multiple of $2^k$ which is at most $t$. Define $D^{(k)}$ to equal $D$ on all blocks which were not updated by $\textsc{UpdateScaling}$ between times $\that$ and $t$, and otherwise set the block to be $0$.

We first argue that we can maintain the matrix $Q_k D^{(k)} A$ in amortized time $\O(d)$ per call to $\textsc{UpdateScaling}$. Indeed, in a call to $\textsc{UpdateScaling}$, one block of $D^{(k)}$ may get set to $0$, which sets at most $M$ rows of $D^{(k)}A$ to be $0$. Because each column of $Q_k$ has $\O(1)$ nonzero entries, we can maintain $Q_k D^{(k)} A$ in $\O(Md)$ time. Now, during a call to $\textsc{UpdateSlack}$ that makes $t$ a multiple of $2^k$, $D^{(k)}$ gets reset to $D$. Because every block of $D^{(k)}$ may get set to $0$ and reset to $D_i$ at most once, the runtime of this step can get charged to $\textsc{UpdateScaling}$.

Now we discuss how to implement the $t$-th call to $\textsc{UpdateSlack}$ for $t \ge 1$. Let $h^{(t})$ be the vector $h$ in the $t$-th call to $\textsc{UpdateSlack}(h)$ and let $s^{(t)}$ be the slack vector. For $k \in \mathbb{Z}_{\ge0}$ such that $2^k \mid t$, call $\textsc{Recover}(v)$ for $v = Q_k D^{(k)} A \sum_{s = t-2^k+1}^t h^{(s)}$ which returns a subset $S_k \subseteq [n]$ containing coordinates $j$ such that
\[ \left|\left(D^{(k)} A \sum_{s = t-2^k+1}^t h^{(s)}\right)_j\right| \ge \beps_k \left\|D^{(k)} A \sum_{s = t-2^k+1}^t h^{(s)} \right\|_2. \]
Now for a block $i$ containing $j \in S_k$, if there were no calls to $\textsc{UpdateScaling}(i, \cdot)$ in the times $[t - 2^k + 1, t]$, and
\begin{equation} \left\|D_i(s^{(t)}_i - s^{(t-2^k+1)}_i)\right\|_2 \ge \beps, \label{eq:checks} \end{equation}
then set $\sbar_i = s^{(t)}_i$. Finally, also update $\sbar_i \assign s^{(t)}_i$ after every call to $\textsc{UpdateScaling}(i, \cdot)$.

\paragraph{Analysis.} We now analyze the algorithm described above. We already described how to maintain the matrices $Q_k D^{(k)}A$ in amortized $\O(d)$ time. Updating $\sbar$ during a call to $\textsc{UpdateScaling}$ also costs $\O(d)$ time.

For $\textsc{UpdateSlack}$ we will bound $|S_k|$ and the time needed to find $S_k$. The matrix $Q_k D^{(k)} A$ is $\O(\beps_k^{-2}) \times d$, so computing
\[ Q_k D^{(k)} A \sum_{s = t-2^k+1}^t h^{(s)} \] costs time $\O(\beps_k^{-2}d) = \O(\eps^{-2}2^{2v_2(t)} d)$ (we can use partial sums to find the vector $\sum_{s = t-2^k+1}^t h^{(s)}$ in $O(d)$ time). Similarly, $|S_k| \le \O(\beps_k^{-2}) = \O(\eps^{-2}2^{2v_2(t)})$ by the guarantees of \cref{thm:hh}.
Thus checking the relevant blocks in \eqref{eq:checks} costs time $\O(d|S_k|) \le \O(\eps^{-2}2^{2v_2(t)} d)$.

All that is left is to verify the correctness of the algorithm. First, note that
\begin{align*}
    \left\|D^{(k)} A \sum_{s = t-2^k+1}^t h^{(s)} \right\|_2 &\le \sum_{s=t-2^k+1} \left\|D^{(k)} Ah^{(s)}\right\|_2 \le 2^k,
\end{align*}
because by definition, on each block either $D^{(k)}$ was not updated in times $[t - 2^k + 1, t]$ or was set to $0$. Thus, $S_k$ contains all coordinates $j$ such that
\[ \left(D(s^{(t)} - s^{(t-2^k)})\right)_j \ge \beps_k 2^k = \frac{\beps}{5M} \] whp, by the guarantees of $\textsc{Recover}$ in \cref{thm:hh}. Thus every block $i$ with $\|D_i(s^{(t)}_i - s^{(t-2^k)})_i\|_2 \ge \frac{\beps}{2}$ is checked in \eqref{eq:checks} because blocks are size $n_i \times n_i$ for $n_i \le M$. This establishes that each $\|D_i(\sbar_i - \sbar)\|_2 \le \eps$ at all times because every interval can be broken up into at most $2 \log n$ intervals of the form $[t-2^k+1, t]$.

The algorithm succeeds against an adaptive adversary because the update sequence of $\sbar$ works against an adaptive adversary because the coordinates it is defined on only depend on the $h^{(t)}$.
\end{proof}

Next we describe the main results we need for maintaining $\xbar$. The key point is to maintain a data structure that can sample a valid block-diagonal matrix $R$. For this, we first need a JL-based algorithm that lets us sample a coordinate of a vector proportional to its contribution to the $\ell_2$-norm. This is based on \cite[Lemma B.3]{BLLSSSW21}, adapted to our setting.
\begin{lemma}
\label{lem:samplel2}
There is a data structure that given a matrix $A \in \R^{n \times d}$ and block-diagonal PSD matrix $D \in \R^{n \times n}$ with blocks $D_i \in \R^{n_i \times n_i}$ for $i \in [m]$ and $M = \max_{i \in [m]} n_i$, initializes in time $\O(nd)$ and supports the following operations.
\begin{itemize}
    \item $\textsc{UpdateScaling}(i, N \in \R^{n_i \times n_i})$. Sets $D_i \assign N$.
    \item $\textsc{Sample}(h \in \R^d)$. Returns a random single blocks $b' \in [m]$ such that for all $b \in [m]$ (corresponding to block $B \subseteq [n]$) it holds that
    \[ \Pr[b' = b] = \frac{\sum_{j \in B} (DAh)_j^2}{\|DAh\|_2^2}. \]
\end{itemize}
The algorithm initializes in time $\O(nd)$ and each operation can be handled in $\O(d)$ time whp.
\end{lemma}
\begin{proof}
Build a binary tree of intervals over the block indices $[m]$ and for an interval $I \subseteq [m]$ let $S_I \subseteq [n]$ be the union of the coordinates in the blocks in $I$, and let $J_I \in \R^{\O(1) \times S_I}$ be a JL matrix. Our algorithm will maintain the matrices $J_I D_I A_I h$ where $D_I \in \R^{S_I \times S_I}$ is the restriction of $D$ to the blocks $I$, and $A_I \in \R^{S_I \times d}$ is the restriction of $A$ to the coordinates in $S_I$. Because $\sum_I |S_I| \le \O(n)$, the time to initialize all the $J_I D_I A_I$ matrices is $\O(nd)$. Also, because each block $i \in [m]$ is only in $O(\log n)$ intervals $I$, the total time to update the matrices $J_I D_I A_I$ during a call to $\textsc{UpdateScaling}$ is bounded by $\O(d)$.

Now we describe how to implement $\textsc{Sample}(h)$. Initialize the interval $I = [m]$. While $I$ is not size $1$, let $I_L$ and $I_R$ be its children and consider the quantities $\|J_ID_IA_Ih\|_2^2 \approx_{1+\eps} \|D_IA_Ih\|_2^2$, $\|J_{I_L}D_{I_L}A_{I_L}h\|_2^2 \approx_{1+\eps} \|D_{I_L}A_{I_L}h\|_2^2$, $\|J_{I_R}D_{I_R}A_{I_R}h\|_2^2 \approx_{1+\eps} \|D_{I_R}A_{I_R}h\|_2^2$ where $\eps \le 1+\frac{1}{100\log n}$. Now, go from $I$ down to $I_L$ with probability
\[ \frac{\|J_{I_L}D_{I_L}A_{I_L}h\|_2^2}{\|J_{I_L}D_{I_L}A_{I_L}h\|_2^2 + \|J_{I_R}D_{I_R}A_{I_R}h\|_2^2} \]
and move to $I_R$ otherwise. Finally, when you get to a single block $i \in [m]$ define
\[ p_i = \frac{\|D_i A_i h\|_2^2}{2 \prod_{I \ni i} \frac{\|J_{I_Y}D_{I_Y}A_{I_Y}h\|_2^2}{\|J_{I_L}D_{I_L}A_{I_L}h\|_2^2 + \|J_{I_R}D_{I_R}A_{I_R}h\|_2^2}} \]
where $Y \in \{L, R\}$ such that $i \in I_Y$. It can be checked that $p_i \le \frac{1}{2}(1+\eps)^{3 \log n} < 1$ and $p_i \ge \frac{1}{2}(1-\eps)^{3 \log n} > 1/4$. Now, return $i$ with probability $p_i$, and otherwise return nothing. If nothing is returned, restart the process. We need at most $\O(1)$ runs with high probability because $p_i > 1/4$. Evidently, each step can be implemented in time $\O(d)$ because computing each $\|J_ID_IA_Ih\|_2^2$ and $\|D_i A_i h\|_2^2$ takes $\O(Md)$ time.
\end{proof}

Finally, we establish that sampling by a combination of (1) proportional to the $\ell_2$-norm of blocks, and (2) uniform, and (3) leverage score overestimates, produces a \emph{valid} block-diagonal matrix $R$, as defined in \cref{def:valid}.

\begin{lemma}
\label{lem:validsampling}
Consider a block-diagonal matrix $D \in \R^{n \times n}$, $A \in \R^{n \times d}$, and $\delta \in \R^n$. For $i \in [m]$ corresponding to block $S_i \subseteq [n]$ let $\tilde{\tau}_i \ge \sum_{j \in S_i} \tau(DA)_j$ and $T = \sum_{i \in [m]} \tilde{\tau}_i$.

Let $K = 2\sqrt{m} + T$. Sample a single $i \in [m]$ with probability
\[ p_i \coloneqq \frac{\sqrt{m}\left(\frac{\|\delta_i\|_2^2}{\|\delta\|_2^2} + \frac{1}{m}\right) + \tilde{\tau}_i}{K}. \]
For a sufficiently large constant $C$ take $K' = C(\alpha\gamma)^{-2}\log n \cdot K$ samples $i_1, \dots, i_{K'}$, and let \[ R = \sum_{j=1}^{K'} \frac{1}{p_{i_j}K'} I_{S_{i_j}}, \] where $I_{S_i}$ is the identity matrix on block $i$. Then $R$ is valid according to \cref{def:valid}.
\end{lemma}
\begin{proof}
The first two items of \cref{def:valid} follow by construction. For item 3 (Variance), let $E_j$ be the $i$-th entry of $\frac{1}{p_{i_j}K'} I_{S_{i_j}}$ so that $\E[E_j] = \frac{1}{K'}$ and $R_{ii} = \sum_{j \le K'} E_j$. Let $b$ be the block containing $i$. Then
\begin{align*}
\Var(R_{ii}) &= \E[R_{ii}^2] - 1 = \sum_j \E[E_j^2] + \sum_{j \neq j'} \E[E_j E_{j'}] - 1 \\
&= K' \frac{1}{(p_b K')^2} p_b + \frac{K'(K'-1)}{(K')^2} - 1 \le \frac{1}{p_b K'}.
\end{align*}
Also, note that
\[ \frac{1}{p_b K'} \ge \frac{K}{K' \cdot \frac{|\delta_i|}{\|\delta\|_2}}, \]
where we have applied the inequality $a+b \ge 2\sqrt{ab}$ to say that
\[ \sqrt{n}\left(\frac{\|\delta_b\|_2^2}{\|\delta\|_2^2} + \frac{1}{m}\right) \ge \frac{|\delta_b\|_2}{\|\delta\|_2} \ge \frac{|\delta_i|}{\|\delta\|_2^2}. \]
Thus
\[ \Var(R_{ii}\delta_i) \le \frac{K}{K' \cdot \frac{|\delta_i|}{\|\delta\|_2}} \cdot \delta_i^2 = \frac{K|\delta_i|\|\delta\|_2}{K'}, \] which completes the proof by the choice of $K'$.
Item 4 (Covariance) follows because $R_{ii}$ and $R_{jj}$ are negatively correlated. Item 5 (Maximum) follows because the maximum value of $E_j\delta_i$ is at most $\frac{1}{p_iK'}\delta_i \le \frac{K\|\delta\|_2}{K'}$, so the result follows from this plus the bound on $\Var(R_{ii}\delta_i)$, and Bernstein's inequality. Finally, item 6 (Spectral approximation) follows by the matrix Bernstein bound (see \cref{lem:sample}) applied to our choice of sampling probabilities $p_i$, which are leverage score overestimates.
\end{proof}
Note that we can sample according to the necessary probabilities $p_i$ as defined in \cref{lem:validsampling} by using \cref{lem:samplel2}.

\subsection{Initial and Final Point}
\label{subsec:initialfinal}

To initialize the IPM with a well-centered point, we can directly use \cite[Lemma D.2]{LSZ19}. To prove that a well-centered point for small path parameter $t$ is approximately optimal, we mimic the proof of \cite[Lemma D.3]{LSZ19} combined with \cite[Lemma 4.11]{BLLSSSW21}.

\begin{lemma}[Final point]
\label{lem:finalpoint}
Given an $\eps$-well-centered point $(x, s)$ for path parameter $t$, we can compute a feasible pair $(x^{(\final)}, s^{(\final)})$ such that:
\begin{enumerate}
    \item $A^\top x^{(\final)} = b$ and $s^{(\final)} = c - Ay$ for some $y \in \R^d$, and
    \item $c^\top x^{(\final)} - \min_{\substack{x \in K_1 \times \dots \times K_m \\ A^\top x = b}} c^\top x \lesssim nt$,
\end{enumerate}
\end{lemma}
\begin{proof}
We set $s^{(\final)} = s$ and $x^{(\final)} = x - \g^2 \Phi(x)^{-1}A(A^\top \g^2 \Phi(x)^{-1} A)^{-1} (A^\top x - b)$. We start by arguing that $x^{(\final)}$ is feasible. Towards this, by standard self-concordance facts (see eg. \cite{Nes98}), it suffices to argue that
\begin{equation} \|\g^2 \Phi(x)^{1/2}(x^{(\final)} - x)\|_{\infty, 2} \le \alpha\eps. \label{eq:closesc} \end{equation}
Indeed, this follows because
\begin{align*}
    \|\g^2 \Phi(x)^{1/2}(x^{(\final)} - x)\|_{\infty, 2} &\le \|\g^2 \Phi(x)^{1/2}(x^{(\final)} - x)\|_2 \\
    &= \|\g^2 \Phi(x)^{-1/2} A(A^\top \g^2 \Phi(x)^{-1} A)^{-1}(A^\top x - b)\|_2 \\
    &= \|A^\top x - b\|_{(A^\top \g^2 \Phi(x)^{-1} A)^{-1}} \le \alpha\eps
\end{align*}
where the final step is because $(x, s)$ is well-centered. Next, we argue that $(x^{(\final)}, s^{(\final)})$ is $5\eps$-well-centered. Indeed, for a block $i \in [m]$, we bound
\begin{align*}
\left\|\frac{s_i}{t} + \g \phi_i(x^{(\final)}_i)\right\|_{\g^2 \phi_i(x^{(\final)}_i)^{-1}} &\le 2 \left\|\frac{s_i}{t} + \g \phi_i(x^{(\final)}_i)\right\|_{\g^2 \phi_i(x_i)^{-1}} \\
&\le 2\gamma_i^t(x, s)^{1/2} + 2\left\|\g \phi_i(x^{(\final)}_i) - \g \phi_i(x_i)\right\|_{\g^2 \phi_i(x_i)^{-1}} \\
&\le 2\eps + 4\alpha\eps,
\end{align*}
where the final line uses standard self-concordance facts, i.e., $\g^2 \phi_i(x_i) \approx_2 \g^2 \phi_i(x^{(\final)}_i)$ and
\[ 2\left\|\g \phi_i(x^{(\final)}_i) - \g \phi_i(x_i)\right\|_{\g^2 \phi_i(x_i)^{-1}} \lesssim \alpha\eps, \]
by \eqref{eq:closesc}. Because $(x^{(\final)}, s^{(\final)})$ are feasible points that are well-centered, by second item now follows by \cite[Lemma D.3]{LSZ19}.
\end{proof}

\subsection{Overall Runtime Analysis}
\label{subsec:overall}

In this section we analyze the runtime of implementing $\O(\sqrt{n})$ iterations of \cref{alg:shortstep}, which will prove our main theorem (\cref{thm:main}). Towards this we need to maintain $\xbar$, $\sbar$, the vector $A^\top x - b$, and sample the sparsifier $H \approx A^\top \g^2 \Phi(\xbar)^{-1} A$ during each iteration.

\begin{proof}[Proof of \cref{thm:main}]
The algorithm is as follows. Initialize the initial program as in \cite[Lemma D.2]{LSZ19}, then run the short-step procedure in \cref{alg:shortstep} for $\O(\sqrt{n})$ steps, and return the final point as described in \cref{lem:finalpoint}. By \cref{lem:pot_drop} it holds whp that the points $(x, s)$ in the algorithm are all $\eps$-well-centered.

Throughout the algorithm is running an instance $\mathcal{D}^{(\lev)}$ of \cref{thm:sparsify} to maintain leverage score overestimates of the matrix $A^\top \g^2 \Phi(\xbar)^{-1} A$. Formally, because $\g^2 \Phi(\xbar)^{-1}$ is block diagonal (with PSD blocks) instead of diagonal as is required by \cref{thm:sparsify}, we need to make a small modification in its implementation. Every time $\xbar_i$ updates for a block $i \in [m]$, pass deletions of all rows $a_i$ corresponding to that block, and pass insertions the following rows to $\mathcal{D}^{(\lev)}$. Let $\g^2 \phi_i(\xbar_i) = U^\top DU \in \R^{n_i \times n_i}$ be the SVD, and insert the rows of the matrix $UA_i$, where $A_i$ is the restriction of $A$ to the $i$-th block.

Next we discuss the maintenance of $\xbar$ and $\sbar$. We will prove inductively that $\xbar$ and $\sbar$ can be maintained in a way where only $\O(n)$ coordinates update ever.

\paragraph{Maintaining $x$ and $\xbar$.} It is useful to discuss how to maintain $x$ and $\xbar$ together. $x$ is maintained implicitly: we maintain $g$ which changes in at most $\O(n)$ coordinates, and maintain running sums. By \cref{lem:validsampling} there is a way to sample a valid matrix $R$ with at most $O(\sqrt{m} + d)$ samples, where the leverage scores are maintained and returned by $\mathcal{D}^{(\lev)}$. The term $R\delta_r$ is handled explicitly, which costs $\O(\sqrt{n} \cdot (\sqrt{n} + d))$ times. By \cref{lem:samplel2}, each sample can be done in time $\O(d)$ plus $\O(nd)$ preprocessing. This allows us to sample $R$ and thus maintain $x$ in time
\[ \O(d \cdot \sqrt{n} \cdot (\sqrt{n} + d)) \le \O(nd + d^2\sqrt{n}) \le \O(nd + d^3). \]
Also, we can maintain $A^\top x - b$ in the same runtime: $O(d)$ per change to a coordinate of $x$, plus the time needed to maintain $A^\top g$, which is $\O(nd)$ total because $g$ changes in at most $\O(n)$ coordinates throughout.

Maintenance of $\xbar$ is done by maintaining changes on the $g$ and $R\delta_r$ terms separately. The data structure can maintain running sums of $\g^2 \Phi(\xbar)^{1/2} g$ and decide when partial sums have accumulated more than $\beta/3$ and use these to update $\xbar$. This happens only $\O(n)$ times, because $\|g\|_2 \le \eps$ at each iteration. Now we discuss how to maintain when accumulations of the $\Phi(\xbar)^{-1/2} R\delta$ terms are large. This is done greedily, which is acceptable for runtime because the vector $\Phi(\xbar)^{-1/2} R\delta$ is $O(\sqrt{n} + d)$-sparse. To argue that this only changes $\O(n)$ coordinates, we apply \cref{lem:xstable}: there is a sequence $\hat{x}$ which is an $\alpha < \beta/10$-approximation to $x$ (see item 1 of \cref{lem:xstable}) which satisfies $\|\g^2 \Phi(x^{(k)})^{1/2}(\hat{x}^{(k+1)} - \hat{x}^{(k)})\|_2 \le \eps$ (this is item 2), so coordinates of $\hat{x}$ only undergo changes of size at least $\eps/10$ at most $\O(n)$ times over $\O(\sqrt{n})$ IPM steps.

\paragraph{Maintaining $\sbar$.} We will prove that we can maintain $\sbar$ to be an $\eps$-approximation of $s$ as in \cref{def:epsapprox} with at most $\O(n)$ total changes. Indeed, we can simply use the data structure in \cref{lem:slack}, along with the fact that $\|\Phi(x)^{1/2}\bar{\delta}_s\|_2 \le t$ by \cref{lem:step_size}.
Because $\xbar$ changes at most $\O(n)$ total times, the running time is at most $\O(nd)$.

\paragraph{Running time of $\mathcal{D}^{(\lev)}$.} The total number of row updates is at most $\O(n)$ and the number of batches is $\O(\sqrt{n})$. Additionally, the condition number of $A^\top \g^2 \Phi(x)^{-1} A$ is lower and upper bounded by $\poly(\kappa n)$ throughout by \cref{lemma:hessianapprox}.
So by \cref{thm:sparsify} the total running time is at most $\O(nd + d^6\sqrt{n}) \le \O(nd + d^{11})$.

The overall runtime is dominated by the running time of the sparsifier, which completes the proof of \cref{thm:main}.

\end{proof}

\section*{Acknowledgments}

Yang P. Liu would like to thank Yin Tat Lee and Aaron Sidford for discussions about ERM.
Richard Peng would like to thank
Chenxin Dai, Alicia Stepin, and Zhizheng Yuan for
discussions related to weighted central paths,
and Michael B. Cohen and Jelani Nelson for various discussions related to dynamic matrix sparsification.
Albert Weng was supported by NSF Award CCF-2338816.

\bibliographystyle{alpha}
\bibliography{refs}

\newcommand{\etalchar}[1]{$^{#1}$}
\begin{thebibliography}{KNPW11}

\bibitem[AJK25]{AJK25}
Deeksha Adil, Shunhua Jiang, and Rasmus Kyng.
\newblock Acceleration meets inverse maintenance: Faster $\ell_{\infty}$-regression.
\newblock In Keren Censor{-}Hillel, Fabrizio Grandoni, Jo{\"{e}}l Ouaknine, and Gabriele Puppis, editors, {\em 52nd International Colloquium on Automata, Languages, and Programming, {ICALP} 2025, July 8-11, 2025, Aarhus, Denmark}, volume 334 of {\em LIPIcs}, pages 5:1--5:16. Schloss Dagstuhl - Leibniz-Zentrum f{\"{u}}r Informatik, 2025.

\bibitem[AKPS19]{AKPS19}
Deeksha Adil, Rasmus Kyng, Richard Peng, and Sushant Sachdeva.
\newblock Iterative refinement for {$\ell_p$}-norm regression.
\newblock In {\em Proceedings of the {T}hirtieth {A}nnual {ACM}-{SIAM} {S}ymposium on {D}iscrete {A}lgorithms}, pages 1405--1424. SIAM, Philadelphia, PA, 2019.

\bibitem[BCLL18]{BCLL18}
S\'ebastien Bubeck, Michael~B. Cohen, Yin~Tat Lee, and Yuanzhi Li.
\newblock An homotopy method for {$\ell_p$} regression provably beyond self-concordance and in input-sparsity time.
\newblock In {\em S{TOC}'18---{P}roceedings of the 50th {A}nnual {ACM} {SIGACT} {S}ymposium on {T}heory of {C}omputing}, pages 1130--1137. ACM, New York, 2018.

\bibitem[BGJ{\etalchar{+}}22]{BGJLLPS22}
Jan van~den Brand, Yu~Gao, Arun Jambulapati, Yin~Tat Lee, Yang~P. Liu, Richard Peng, and Aaron Sidford.
\newblock Faster maxflow via improved dynamic spectral vertex sparsifiers.
\newblock In Stefano Leonardi and Anupam Gupta, editors, {\em {STOC} '22: 54th Annual {ACM} {SIGACT} Symposium on Theory of Computing, Rome, Italy, June 20 - 24, 2022}, pages 543--556. {ACM}, 2022.

\bibitem[BGKS23]{BVKS23:journal}
Jess Banks, Jorge {Garza-Vargas}, Archit Kulkarni, and Nikhil Srivastava.
\newblock Pseudospectral shattering, the sign function, and diagonalization in nearly matrix multiplication time.
\newblock {\em Foundations of computational mathematics}, 23(6):1959--2047, 2023.

\bibitem[BLL{\etalchar{+}}21]{BLLSSSW21}
Jan van~den Brand, Yin~Tat Lee, Yang~P. Liu, Thatchaphol Saranurak, Aaron Sidford, Zhao Song, and Di~Wang.
\newblock Minimum cost flows, mdps, and $\ell_1$-regression in nearly linear time for dense instances.
\newblock In Samir Khuller and Virginia~Vassilevska Williams, editors, {\em {STOC} '21: 53rd Annual {ACM} {SIGACT} Symposium on Theory of Computing, Virtual Event, Italy, June 21-25, 2021}, pages 859--869. {ACM}, 2021.

\bibitem[BLN{\etalchar{+}}20]{BLNPSSSW20}
Jan van~den Brand, Yin~Tat Lee, Danupon Nanongkai, Richard Peng, Thatchaphol Saranurak, Aaron Sidford, Zhao Song, and Di~Wang.
\newblock Bipartite matching in nearly-linear time on moderately dense graphs.
\newblock In Sandy Irani, editor, {\em 61st {IEEE} Annual Symposium on Foundations of Computer Science, {FOCS} 2020, Durham, NC, USA, November 16-19, 2020}, pages 919--930. {IEEE}, 2020.

\bibitem[BLSS20]{BLSS20}
Jan van~den Brand, Yin~Tat Lee, Aaron Sidford, and Zhao Song.
\newblock Solving tall dense linear programs in nearly linear time.
\newblock In Konstantin Makarychev, Yury Makarychev, Madhur Tulsiani, Gautam Kamath, and Julia Chuzhoy, editors, {\em Proceedings of the 52nd Annual {ACM} {SIGACT} Symposium on Theory of Computing, {STOC} 2020, Chicago, IL, USA, June 22-26, 2020}, pages 775--788. {ACM}, 2020.

\bibitem[Bra20]{Brand20}
Jan van~den Brand.
\newblock A deterministic linear program solver in current matrix multiplication time.
\newblock In {\em Proceedings of the 2020 {ACM}-{SIAM} {S}ymposium on {D}iscrete {A}lgorithms}, pages 259--278. SIAM, Philadelphia, PA, 2020.

\bibitem[Bra21]{Brand21}
Jan van~den Brand.
\newblock Unifying matrix data structures: Simplifying and speeding up iterative algorithms.
\newblock In Hung~Viet Le and Valerie King, editors, {\em 4th Symposium on Simplicity in Algorithms, {SOSA} 2021, Virtual Conference, January 11-12, 2021}, pages 1--13. {SIAM}, 2021.

\bibitem[BZ23]{BZ23}
Jan van~den Brand and Daniel~J. Zhang.
\newblock Faster high accuracy multi-commodity flow from single-commodity techniques.
\newblock In {\em 64th {IEEE} Annual Symposium on Foundations of Computer Science, {FOCS} 2023, Santa Cruz, CA, USA, November 6-9, 2023}, pages 493--502. {IEEE}, 2023.

\bibitem[Che21]{Chewi21}
Sinho Chewi.
\newblock The entropic barrier is n-self-concordant.
\newblock {\em CoRR}, abs/2112.10947, 2021.

\bibitem[CKL{\etalchar{+}}25]{CKLPPS25}
Li~Chen, Rasmus Kyng, Yang Liu, Richard Peng, Maximilian {Probst Gutenberg}, and Sushant Sachdeva.
\newblock Maximum flow and minimum-cost flow in almost-linear time.
\newblock {\em J. ACM}, 72(3):Art. 19, 103, 2025.

\bibitem[CKM{\etalchar{+}}11]{CKMST11}
Paul Christiano, Jonathan~A. Kelner, Aleksander Madry, Daniel~A. Spielman, and Shang-Hua Teng.
\newblock Electrical flows, {L}aplacian systems, and faster approximation of maximum flow in undirected graphs.
\newblock In {\em {STOC}}, 2011.

\bibitem[Cla05]{C05}
Kenneth~L. Clarkson.
\newblock Subgradient and sampling algorithms for $\ell_1$ regression.
\newblock In {\em Proceedings of the {S}ixteenth {A}nnual {ACM}-{SIAM} {S}ymposium on {D}iscrete {A}lgorithms}, pages 257--266. ACM, New York, 2005.

\bibitem[CLS21]{CLS19}
Michael~B. Cohen, Yin~Tat Lee, and Zhao Song.
\newblock Solving linear programs in the current matrix multiplication time.
\newblock {\em J. ACM}, 68(1):Art. 3, 39, 2021.

\bibitem[CMP20]{CMP20}
Michael~B. Cohen, Cameron Musco, and Jakub Pachocki.
\newblock Online row sampling.
\newblock {\em Theory Comput.}, 16:1--25, 2020.

\bibitem[Cox58]{Cox58}
D.~R. Cox.
\newblock The regression analysis of binary sequences.
\newblock {\em J. Roy. Statist. Soc. Ser. B}, 20:215--242, 1958.

\bibitem[CV95]{CV95}
Corinna Cortes and Vladimir Vapnik.
\newblock Support-vector networks.
\newblock {\em Machine Learning}, 20(3):273--297, 1995.

\bibitem[DDH{\etalchar{+}}08]{DDHKM08}
Anirban Dasgupta, Petros Drineas, Boulos Harb, Ravi Kumar, and Michael~W. Mahoney.
\newblock Sampling algorithms and coresets for {$\ell_p$} regression.
\newblock In {\em Proceedings of the {N}ineteenth {A}nnual {ACM}-{SIAM} {S}ymposium on {D}iscrete {A}lgorithms}, pages 932--941. ACM, New York, 2008.

\bibitem[FMP{\etalchar{+}}18]{FMPSWX18}
Matthew Fahrbach, Gary~L. Miller, Richard Peng, Saurabh Sawlani, Junxing Wang, and Shen~Chen Xu.
\newblock Graph sketching against adaptive adversaries applied to the minimum degree algorithm.
\newblock In Mikkel Thorup, editor, {\em 59th {IEEE} Annual Symposium on Foundations of Computer Science, {FOCS} 2018, Paris, France, October 7-9, 2018}, pages 101--112. {IEEE} Computer Society, 2018.

\bibitem[FS97]{FS97}
Yoav Freund and Robert~E. Schapire.
\newblock A decision-theoretic generalization of on-line learning and an application to boosting.
\newblock volume~55, pages 119--139. 1997.
\newblock Second Annual European Conference on Computational Learning Theory (EuroCOLT '95) (Barcelona, 1995).

\bibitem[GLP21]{GLP21:journal}
Yu~Gao, Yang Liu, and Richard Peng.
\newblock Fully dynamic electrical flows: Sparse maxflow faster than goldberg–rao.
\newblock {\em SIAM Journal on Computing}, 0(0):FOCS21--85--FOCS21--156, 2021.

\bibitem[HLS13]{HLS13}
David~W. Hosmer, Stanley Lemeshow, and Rodney~X. Sturdivant.
\newblock {\em Applied Logistic Regression}.
\newblock John Wiley \& Sons, Hoboken, NJ, 3 edition, 2013.

\bibitem[JL{\etalchar{+}}84]{JL84}
William~B Johnson, Joram Lindenstrauss, et~al.
\newblock Extensions of lipschitz mappings into a hilbert space.
\newblock {\em Contemporary mathematics}, 26(189-206):1, 1984.

\bibitem[JSWZ21]{JSWZ21}
Shunhua Jiang, Zhao Song, Omri Weinstein, and Hengjie Zhang.
\newblock A faster algorithm for solving general lps.
\newblock In Samir Khuller and Virginia~Vassilevska Williams, editors, {\em {STOC} '21: 53rd Annual {ACM} {SIGACT} Symposium on Theory of Computing, Virtual Event, Italy, June 21-25, 2021}, pages 823--832. {ACM}, 2021.

\bibitem[KH01]{KH01}
Roger Koenker and Kevin~F. Hallock.
\newblock Quantile regression.
\newblock {\em Journal of Economic Perspectives}, 15(4):143--156, 2001.

\bibitem[KNPW11]{KNPW11}
Daniel~M. Kane, Jelani Nelson, Ely Porat, and David~P. Woodruff.
\newblock Fast moment estimation in data streams in optimal space.
\newblock In {\em Proceedings of the 43rd {ACM} Symposium on Theory of Computing, {STOC} 2011, San Jose, CA, USA, June 6-8 2011}, pages 745--754. {ACM}, 2011.
\newblock Available at~\url{https://arxiv.org/abs/1007.4191}.

\bibitem[Koe00]{K00}
Roger Koenker.
\newblock Galton, edgeworth, frisch, and prospects for quantile regression in econometrics.
\newblock {\em Journal of Econometrics}, 95(2):347--374, 2000.

\bibitem[LS14]{LS14}
Yin~Tat Lee and Aaron Sidford.
\newblock Path finding methods for linear programming: Solving linear programs in $\widetilde{O}(\sqrt{\mathsf{rank}})$ iterations and faster algorithms for maximum flow.
\newblock In {\em 55th {IEEE} Annual Symposium on Foundations of Computer Science, {FOCS} 2014, Philadelphia, PA, USA, October 18-21, 2014}, pages 424--433. {IEEE} Computer Society, 2014.

\bibitem[LS15]{LS15}
Yin~Tat Lee and Aaron Sidford.
\newblock Efficient inverse maintenance and faster algorithms for linear programming.
\newblock In Venkatesan Guruswami, editor, {\em {IEEE} 56th Annual Symposium on Foundations of Computer Science, {FOCS} 2015, Berkeley, CA, USA, 17-20 October, 2015}, pages 230--249. {IEEE} Computer Society, 2015.

\bibitem[LSZ19]{LSZ19}
Yin~Tat Lee, Zhao Song, and Qiuyi Zhang.
\newblock Solving empirical risk minimization in the current matrix multiplication time.
\newblock In Alina Beygelzimer and Daniel Hsu, editors, {\em Conference on Learning Theory, {COLT} 2019, 25-28 June 2019, Phoenix, AZ, {USA}}, volume~99 of {\em Proceedings of Machine Learning Research}, pages 2140--2157. {PMLR}, 2019.

\bibitem[LY21]{LY21}
Yin~Tat Lee and Man{-}Chung Yue.
\newblock Universal barrier is \emph{n}-self-concordant.
\newblock {\em Math. Oper. Res.}, 46(3):1129--1148, 2021.

\bibitem[Nes98]{Nes98}
Yurii Nesterov.
\newblock Introductory lectures on convex programming volume i: Basic course.
\newblock {\em Lecture notes}, 3(4):5, 1998.

\bibitem[NN94]{NN94}
Yurii~E. Nesterov and Arkadii Nemirovskii.
\newblock {\em Interior-point polynomial algorithms in convex programming}, volume~13 of {\em Siam studies in applied mathematics}.
\newblock {SIAM}, 1994.

\bibitem[Ren88]{Ren88}
James Renegar.
\newblock A polynomial-time algorithm, based on {N}ewton's method, for linear programming.
\newblock {\em Math. Programming}, 40(1):59--93, 1988.

\bibitem[Tib96]{Tib96}
Robert Tibshirani.
\newblock Regression shrinkage and selection via the lasso.
\newblock {\em Journal of the Royal Statistical Society: Series B (Methodological)}, 58(1):267--288, 1996.

\bibitem[Vai89]{Vaidya89}
Pravin~M. Vaidya.
\newblock Speeding-up linear programming using fast matrix multiplication (extended abstract).
\newblock In {\em 30th Annual Symposium on Foundations of Computer Science, Research Triangle Park, North Carolina, USA, 30 October - 1 November 1989}, pages 332--337. {IEEE} Computer Society, 1989.

\end{thebibliography}

\newpage

\appendix

\section{Deferred Proofs from Preliminaries}
\label{sec:deferred}

\begin{proof}[Proof of \cref{lem:heavy-hitter}]
Since $\| a_i \|_M^2 = \|M^{1/2} a_i\|_2^2$,
we can compute an $O(\log{n})$-by-$d$ random projection matrix $S$
and obtain the matrix $N \coloneqq M^{1/2} S^{\top}$
in $\Otil(d^{\omega})$ time via nearly matrix multiplication time 
diagonalization methods~\cite{BVKS23:journal},
after which we are looking for the rows with large norms in the
$n \times O(\log{n})$ matrix $AN$.
Note that the dimensions of $N$ also means that we can compute each
row of $AN$ in $\Otil(d)$ time.

To invoke the matrix given in \cref{thm:hh},
first note that we can left-multiply $AN$ by another random
projection $R \in \R^{O(\log{n}) \times n}$ so that
\[
\norm{ RAN }_F^2
\approx_{0.1} \norm{ AN}_{F}
\approx_{0.1} \sum_{i} \norm{a_i}_{M}^2
\]
so we can rescale $M$ so that $\sum_{i} \norm{a_i}_M^2 \approx 1$,
and consider a matrix $Q$ given by \cref{thm:hh} with
\[
\epshh \leftarrow 0.01 \sqrt{\frac{\delta}{ \sum_{i} \norm{a_i}_{M}^2}}.
\]
For a row of $i$ to have $\|a_i^{\top} N\|_2 > \epshh$,
at least some coordinate of it must have magnitude
more than $\epshh / O(\log{n})$.
So we can identify all such rows by calling
$\textsc{Recover}(Q A N_{:, j})$ for each of the $O(\log{n})$ columns of $N$ separately.
The resulting $\O( \delta^{-1}  \sum_{i} \norm{a_i}_{M}^2)$
row indices can then be checked in $\Otil(d)$ time each,
giving the runtime for $\textsc{Query}$.

The runtime of Initialize and update is then the cost of computing
and maintaining $QA$ for sufficiently many values of $\epshh$
to ensure constant factor approximation for any relative threshold value.
We can create one such copy per each $\epshh = 0.9^i$,
and maintain the values $QA$ as $A$ get updated.
Both the initializtion and update cossts then follow from 
the $O(\log^3 n)$-nonzeros per column of $Q$,
and there only being $O(\log{n})$ different values of $\epshh$
due to the cost of the $\epshh < 1/n^{10}$ case exceeding
that of running brute force on all rows of $A$.
\end{proof}

\section{ERM Duality via. Convex Conjugates}
\label{sec:ERMDuality}

Let $f_i^*$ denote the convex conjugate of $f_i$, defined as \[f_i^*(x^*)=\sup_{x\in\R^{n_i}}\langle x^*,x\rangle-f_i(x),\] which is convex as it is the supremum of linear functions. Let $x_i \in \R^{n_i}$ be new variables. We can write \eqref{eq:erm1} as
\begin{align*} \min_{y \in \R^d}\sum_{i=1}^m
f_i\left(A_i y - c_i\right) &= \min_{y \in \R^d} \max_{\substack{x_1, \dots, x_m : x_i \in \R^{n_i}}} \sum_{i=1}^m x_i^\top(A_i y - c_i) - f_i^*(x_i) \\
&= \max_{\substack{x_1, \dots, x_m : x_i \in \R^{n_i}}}
\min_{y \in \R^d} \sum_{i=1}^m x_i^\top \left(A_i y - c_i\right) - f_i^*(x_i) \\
&= \max_{x \in \R^n, A^\top x = 0} \sum_{i=1}^m -c_i^\top x_i - f_i^*(x_i)
= -\min_{x \in \R^n, A^\top x = 0} \sum_{i=1}^m c_i^\top x_i + f_i^*(x_i).
\end{align*}

\section{Central Path Stability}
\label{sec:cpstable}

In this section we establish that the Hessian matrices of a path following IPM are stable along the central path. We start by noting standard properties of self-concordant functions, largely citing facts from a book of Nesterov \cite{Nes98}. Specifically, we cite Theorems 4.1.5, 4.1.7, and 4.2.5.

\begin{lemma}
\label{lemma:cpfacts}
Let $K$ be a convex, compact subset of $\R^n$.
Let $\Phi: \mathsf{int}(K) \to \R$ be a $\nu$-self-concordant barrier function. Then:
\begin{enumerate}
    \item \label{item:containment} If $x \in K$ and $\|y-x\|_{\g^2 \Phi(x)} < 1$ then $y \in K$, and
    \item \label{item:gradbound} If $x, y \in K$ then
    $\l \g \Phi(y) - \g \Phi(x), y - x \r
    \ge
    \frac{\|y-x\|_{\g^2 \Phi(x)}^2}{1 + \|y-x\|_{\g^2 \Phi(x)}}$, and
    \item \label{item:normbound} If $x \in K$ and $u \in \R^n$ such that $x+u, x-u \in K$ it holds that $\|u\|_{\g^2 \Phi(x)} \le \nu + 2\sqrt{\nu}$.
\end{enumerate}
\end{lemma}

From here we come to a key claim: if $\l \g \Phi(y) - \g \Phi(x), y - x \r$ is bounded then we can spectrally relate $\g^2 \Phi(x)$ and $\g^2 \Phi(y)$.

\begin{lemma}
\label{lemma:hessianrelate}
Let $K$ be a convex set and
$\Phi: \mathsf{int}(K) \to \R$ be a $\nu$-self-concordant function.
For any parameter $M \geq 1$ any $x, y \in K$ with
\[
\l \Phi(y) - \g \Phi(x), y - x \r
\leq
M,
\]
we have
\[
h^{-1} \g^2 \Phi(x) \pe \g^2 \Phi(y) \pe h \g^2 \Phi(x)
\]
for some $h = O(\nu M)$.
\end{lemma}

\begin{proof}
We first show the assumptions imply $\|y-x\|_{\g^2 \Phi(x)} \le 2M$.
By \cref{lemma:cpfacts}~\cref{item:gradbound}
we get that the given condition implies,
when $\|y-x\|_{\g^2 \Phi(x)} \geq 1$,
\[
M \geq
\frac{\|y-x\|_{\g^2 \Phi(x)}^2}{1 + \|y-x\|_{\g^2 \Phi(x)}}
\geq
\frac{\|y-x\|_{\g^2 \Phi(x)}^2}{2\|y-x\|_{\g^2 \Phi(x)}}
=
\frac{1}{2} \|y-x\|_{\g^2 \Phi(x)}.
\]
When $\|y-x\|_{\g^2 \Phi(x)} \leq 1$,
the assumption of $M \geq 1$ also implies
$\|y-x\|_{\g^2 \Phi(x)} \le 2M$.

Let
\begin{align*}
B_x &= \left\{z : \norm{z-x}_{\g^2 \Phi\left(x\right)} \le 1 \right\},\\
B_y &= \left\{z : \norm{z-y}_{\g^2 \Phi\left(y\right)} \le 1 \right\},
\end{align*}
be the Hessian balls around $x$ and $y$ respectively.
By \cref{lemma:cpfacts}~\cref{item:containment}
we know that $B_x, B_y \subseteq K$.

We will prove that $B_y \subseteq x + h(B_x - x)$,
i.e., $y$ is contained in a dilation of the Hessian ball around $x$.
This would imply that $\g^2 \Phi(y) \pe h \g^2 \Phi(x)$,
which completes the proof by symmetry

We first create a point past $x$ on the line from $y$ to $x$
which is in $B_x$, and thus $K$.
Let
\[
z \coloneq x - \frac{1}{4M} \left(y-x\right)
\]
This is the homothety center of $y$ w.r.t. $x$,
and the direction was chosen to ensure that
$\|z-x\|_{\g^2 \Phi(x)} < \| y - x \|_{\g^2 \Phi(x)} / (4M) < 1$,
and thus $z \in B_x$, and in turn $z \in K$.

Let $u$ be a step in $y$'s Hessian ball,
aka. $y+u, y-u \in B_y$.
Direct algebraic manipulations give
\begin{align*}
\left( 1 - \frac{1}{1 + 4M} \right) z
+ \frac{1}{1 + 4M} \left( y \pm u \right)
&=
\frac{4M}{1 + 4M}
\left(x - \frac{1}{4M}\left(y-x\right)\right)
+ \frac{1}{1 + 4M} \left(y \pm u\right)
\\ &=
x \pm \frac{1}{1 + 4M} u.
\end{align*}

Thus $x \pm u / (1 + 4M)$ can be expressed as a linear
combination of $z$ and $y \pm u$.
As both $z$ and $y \pm u$ are both in $K$,
the convexity of $K$ gives that $x \pm u/(1 + 4M) \in K$.
So \cref{lemma:cpfacts}~\cref{item:normbound} implies that
for all $u$ such that $y + u \in B_y$, we have
\[
\norm{\frac{u}{1 + 4M}}_{\nabla^2 \Phi\left( x \right)}
\leq
\nu + 2 \sqrt{\nu},
\]
or equivalently
$\|u\|_{\nabla^2 \Phi(x)} \leq (1 + 4M)(\nu + 2 \sqrt{\nu})$.

This can be incorporated back into bounding $y \pm u - x$
by triangle inequality:
\[
\|(y\pm u)-x\|_{\g^2 \Phi(x)}
    \le
    \|y-x\|_{\g^2 \Phi(x)} + \|u\|_{\g^2 \Phi(x)}
    \le
    2M + O\left(M (\nu + 2\sqrt{\nu})\right)
    \le
    O(\nu M),
\]
where the first part of the second inequality is from the
bound obtained at the start of this proof.
Thus we have $B_y \subseteq x + h(B_x - x)$ for $h = O(\nu M)$.
\end{proof}

Finally we use the above bound to establish a relationship between the Hessians of points along the robust central path.
Here recall that $C_K$ is the maximum self-concordance parameter
among the $m$ functions.

\begin{lemma}
\label{lemma:hessianapprox}
Let $(x, s)$ and  $(\xhat, \shat)$ be
$\eps$-well-centered solutions for $\eps < 1/1000$
and path parameters $t$, $\that$.
Then for $r = \max\{\that/t, t/\that\}$
and $h = O_{C_K}(nr)$ it holds that
\[
\frac{\g^2 \Phi_i\left(x_i\right)}{h}
\pe
\g^2 \Phi_i(\xhat_i)
\pe
h \g^2 \Phi_i(x_i)
\]
for all $i \in [m]$.
\end{lemma}
\begin{proof}
\cref{lem:finalpoint} gives that for points that are well centered,
there are nearby points that are completely feasible.
So we may assume that $x$ and $\xhat$ are feasible,
and $(x^{(i)}, s^{(i)})$ are $\eps$-centered for $\eps = 1/200$.

Consider the centering errors
\[
\Delta_i
\coloneq
\frac{s_i}{t} + \g \Phi_i(x_i)
\qquad \text{ and } \qquad
\Deltahat_i
\coloneq
\frac{\shat_i}{\that} + \g \Phi_i(\xhat_i)
\]
for all $i \in [m]$.

We can compute for $i \in [m]$ that
\[
\l \g\Phi_i(x_i), x_i - \xhat_i \r
=
\l \Delta_i, x_i - \xhat_i \r
-
\frac{1}{t} \l s_i, x_i - \xhat_i \r
\le
\eps\norm{x_i - \xhat_i}_{\g^2 \Phi_i(x_i)}
- \frac{1}{t} \l s_i, x_i - \xhat_i \r,
\]
where we used Cauchy-Schwarz inequality on the first dot product
and $\| \Delta_i \|_{\g^2 \Phi_i(x_i)^{-1}} \le \eps$ implied by the $\epsilon$-centeredness of $(x, s)$.
Similarly,
\[
\l -\g\Phi_i(\xhat_i), x_i - \xhat_i \r
\le \eps\norm{x_i - \xhat_i}_{\g^2 \Phi_i(\xhat_i)}
+ \frac{1}{\that} \l \shat_i, x_i - \xhat_i \r.
\]
We will sum these conditions to create the overall dot
products between gradient difference and $x - \xhat$.
First, consider the trailing terms,
$\sum_{i} \l s_i, x_i - \xhat_i\r = \l s, x - \xhat \r$
and $\sum_{i} \l \shat_i, x_i - \xhat_i \r = \l \shat, x - \xhat \r$.
Since $A^\tp x = A^\tp \xhat = b$,
we have
$A^\tp (x - \xhat) = 0$,
so writing out $s$ as $c$ adjusted by a vector in the column space of $A$,
$s = c - A y$,
allows us to simplify this dot product to be just in $c$:
\[
\l s, x - \xhat \r
=
\l c, x - \xhat \r
- y^\tp A^\tp \left( x - \xhat \right)
=
\l c, x - \xhat \r
\]
and similarly $\l \shat, x - \xhat \r = \l c, x - \xhat \r$.

So summing the per function conditions over all $i$ gives us:
\begin{multline*}
\sum_{i\in[m]} \left\l \g\Phi_i(x_i) - \g\Phi_i(\xhat_i), x_i - \xhat_i \right \r\\
\le
\eps\sum_{i\in[m]}
\left(\norm{x_i - \xhat_i}_{\g^2\Phi_i(x_i)}
+ \norm{x_i - \xhat_i}_{\g^2 \Phi_i(\xhat_i)}\right)
+ \abs{\Big(\frac{1}{t} - \frac{1}{\that}\Big)\l c, x - \xhat\r}.
\end{multline*}

Now by \cref{lemma:cpfacts}~\cref{item:gradbound} we get that
\[
\left \l \g\Phi_i(x_i) - \g\Phi_i(\xhat_i), x_i - \xhat_i \right \r
\ge
\frac{\norm{x_i - \xhat_i}_{\g^2 \Phi_i(x_i)}^2}
{1+\norm{x_i - \xhat_i}_{\g^2 \Phi_i(x_i)}}
\ge
\norm{x_i - \xhat_i}_{\g^2 \Phi_i(x_i)} - 1.
\]
Combining this with the above gives
\[
\left(1-2\eps\right) \sum_{i\in[m]}
\l \g\Phi_i(x_i) - \g\Phi_i(\xhat_i), x_i - \xhat_i \r
\le
2\eps m + \abs{\Big(\frac{1}{t} - \frac{1}{\that}\Big) \l c, x - \xhat\r}.
\]

By \cref{lem:finalpoint} we know that
\[
\abs{\l c, x - \xhat\r} \le O_{C_K}\left(n \max\{t, \that\}\}\right)
\]
and thus
\[
\left|\Big(\frac{1}{t} - \frac{1}{\that}\Big)\l c,
x - \xhat\r\right| \le O(nr). \]
Thus
\[
\sum_{i\in[m]}
\l \g\Phi_i(x_i) - \g\Phi_i(\xhat_i), x_i - \xhat_i \r
\le
O_{C_K}\left(nr\right).
\]
The desired result then follows from applying \cref{lemma:hessianrelate} for each $i \in [m]$.
Note that the terms on the LHS are all nonnegative due to convexity of the $\Phi_i$'s,
and the requirement of $M \geq 1$ can be satisfied while increasing the overall bound by at most $m \leq O(n)$.
\end{proof}

\end{document}